\documentclass[sigconf, nonacm, pdfa]{acmart}

\usepackage[a-2b]{pdfx}

\usepackage{colortbl}
\usepackage[english]{babel}
\usepackage{svg}
\usepackage{balance}
\usepackage{microtype}
\usepackage[commandnameprefix=ifneeded, addedmarkup=colored,commentmarkup=footnote,final]{changes}%

\usepackage{tabularx}
\usepackage{ragged2e}
\usepackage{placeins}
\usepackage{stackengine}
\usepackage{enumitem}
\usepackage[longend,linesnumbered,vlined]{algorithm2e}
\usepackage{mathtools}
\RestyleAlgo{ruled}
\SetKwComment{Comment}{/* }{ */}

\newcommand{\Tmulti}{\mathbf{T}}
\newcommand{\subTmulti}[1]{\mathbf{T}_{#1}}
\newcommand{\subTa}{\subTmulti{a}}
\newcommand{\subTb}{\subTmulti{b}}
\newcommand{\subTsingle}[2]{\mathbf{T}_{#1}^{#2}}
\newcommand{\subTaf}{\subTsingle{a}{f}}
\newcommand{\subTbf}{\subTsingle{b}{f}}
\newcommand{\dist}[1]{\operatorname{dist}\left(#1\right)}
\newcommand{\distd}[1]{\operatorname{dist}_{d}\left(#1\right)}
\newcommand{\distdmax}[1]{\operatorname{dist}^{\max}_{d}\left(#1\right)}
\newcommand{\dims}[1]{\operatorname{dims}_{d}\left(#1\right)}

\newcommand{\Let}[2]{#1 $\leftarrow$ #2}

\DeclareMathOperator*{\argmin}{arg\,min}

\usepackage{contour}
\usepackage[normalem]{ulem}

\contourlength{0.2ex}

\definecolor{red}{HTML}{EE4B2B}

\usepackage{hyperref}
\usepackage{cleveref}
\newcommand\vldbdoi{XX.XX/XXX.XX}
\newcommand\vldbpages{XXX-XXX}
\newcommand\vldbvolume{18}
\newcommand\vldbissue{10}
\newcommand\vldbyear{2025}
\newcommand\vldbauthors{\authors}
\newcommand\vldbtitle{\shorttitle} 
\newcommand\vldbavailabilityurl{https://github.com/aidaLabDEI/MOMENTI-motifs}
\newcommand\vldbpagestyle{empty}

\theoremstyle{definition}
\newtheorem{definition}{Definition}[section]

\theoremstyle{remark}
\newtheorem{corollaryinner}{Corollary} %
\newenvironment{corollary}[1][]
 {%
  \if\relax\detokenize{#1}\relax
  \else
    \ifcsname #1-used\endcsname
      \expandafter\xdef\csname #1-used\endcsname{\the\numexpr\csname #1-used\endcsname+1}%
    \else
      \expandafter\gdef\csname #1-used\endcsname{1}%
    \fi
    \renewcommand{\thecorollaryinner}{\ref{#1}.\csname #1-used\endcsname}%
  \fi
  \corollaryinner
 }
 {\endcorollaryinner}

\begin{document}

\title{MOMENTI: Scalable Motif Mining  in~Multidimensional~Time~Series}

\author{Matteo Ceccarello}\authornote{Corresponding author.}
\orcid{0000-0003-2783-0218}
\affiliation{%
  \institution{University of Padova}
  \streetaddress{Via Giovanni Gradenigo}
  \city{Padova}
  \state{Italy}
}
\email{matteo.ceccarello@unipd.it}

\author{Francesco Pio Monaco}
\orcid{0009-0003-6341-3021}
\affiliation{%
  \institution{University of Padova}
  \streetaddress{Via Giovanni Gradenigo}
  \city{Padova}
  \state{Italy}
}
\email{monacofran@dei.unipd.it}

\author{Francesco Silvestri}
\orcid{0000-0002-9077-9921}
\affiliation{%
  \institution{University of Padova}
  \streetaddress{Via Giovanni Gradenigo}
  \city{Padova}
  \state{Italy}
}
\email{francesco.silvestri@unipd.it}

\begin{abstract}
Time series play a fundamental role in many domains, capturing a plethora of information about the underlying data-generating processes.
When a process generates multiple \emph{synchronized} signals we are faced with \emph{multidimensional} time series.
In this context a fundamental problem is that of \emph{motif mining}, where we seek patterns repeating twice with minor variations, spanning some of the dimensions.
State of the art exact solutions for this problem run in time quadratic in the length of the input time series.

We provide a scalable method to find the top-$k$ motifs in multidimensional time series
with probabilistic guarantees on the quality of the results.
Our algorithm runs in subquadratic time in the length of the input, and returns the exact solution with probability at least $1-\delta$, where $\delta$ is a user-defined parameter.
The algorithm is designed to be \emph{adaptive} to the input distribution, self-tuning its parameters while respecting user-defined limits on the memory to use.

Our theoretical analysis is complemented by an extensive experimental evaluation, showing that our algorithm is orders of magnitude faster than the state of the art.

\end{abstract}

\maketitle

\pagestyle{\vldbpagestyle}
\begingroup\small\noindent\raggedright\textbf{PVLDB Reference Format:}\\
\vldbauthors. \vldbtitle. PVLDB, \vldbvolume(\vldbissue): \vldbpages, \vldbyear.\\
\href{https://doi.org/\vldbdoi}{doi:\vldbdoi}
\endgroup
\begingroup
\renewcommand\thefootnote{}\footnote{\noindent
This work is licensed under the Creative Commons BY-NC-ND 4.0 International License. Visit \url{https://creativecommons.org/licenses/by-nc-nd/4.0/} to view a copy of this license. For any use beyond those covered by this license, obtain permission by emailing \href{mailto:info@vldb.org}{info@vldb.org}. Copyright is held by the owner/author(s). Publication rights licensed to the VLDB Endowment. \\
\raggedright Proceedings of the VLDB Endowment, Vol. \vldbvolume, No. \vldbissue\ %
ISSN 2150-8097. \\
\href{https://doi.org/\vldbdoi}{doi:\vldbdoi} \\
}\addtocounter{footnote}{-1}\endgroup

\ifdefempty{\vldbavailabilityurl}{}{
\vspace{.3cm}
\begingroup\small\noindent\raggedright\textbf{PVLDB Artifact Availability:}\\
The source code, data, and/or other artifacts have been made available at \url{\vldbavailabilityurl}.
\endgroup
}

\section{Introduction}
Time series play a central role in modeling the evolution of data-generating processes in many domains.
To capture the multifaceted nature of processes generating data, time series are oftentimes \emph{multidimensional}: collections of co-evolving signals whose measurements are synchronized, collectively describing the evolution of the process.
Extracting information from such multidimensional time series is thus fundamental.

In particular, \emph{top-$k$ motifs} mining is a crucial and challenging problem: intuitively, the goal is to find patterns that occur twice with minor modifications, spanning many, but not all, the signals of the time series.
More formally, given a time series $\mathbf{T}$ with $D>1$ dimensions and $n$ points, the problem consists in finding $k$ pairs of $\mathbf{T}$'s subsequences with the smallest distance, where the distance captures the similarity between subsequences.
Indeed, similar patterns might  imply a particular behavior, making motif discovery a crucial step for higher-level analysis. Applications include forecasts for volcanic eruptions \cite{2012_Cassisi}, healthcare management \cite{health2015}, and machine management in industry \cite{renard:tel-01922186}.
\added{
In particular, multidimensional motif discovery is key in
pollution control~\cite{PollutionMotifs},
in quality control in industrial settings~\cite{SteelMotifs},
in activity discovery~\cite{ActivityDiscovery, ActivityDiscovery2},
and in healthcare~\cite{Healthcare, health2015}.
}

\added{A common approach to motif discovery is to extend the approach for motif mining in 1-dimension time series to the multidimensional case: that is, we look for pairs of subsequences of the time series where all the $D$ dimensions are similar.
However, this approach might not reveal interesting patterns: for instance, some dimensions might be noisy or uncorrelated with respect to the others, and they might hide similar patterns involving only a subset of the dimensions that are in general unknown beforehand.}
\added{ We aim to discover patterns that involve only a subset of the $D$ dimensions which is unknown and needs to be retrieved as well. This formulation overcomes the limitation of several previous approximate approaches that assumed all dimensions as equally informative.
}

Multidimensional motifs can then be discovered with the following naive approach:
for each of the $2^D$ subset of dimensions, we compare the $O(n^2)$ subsequences of $\mathbf{T}$ on the selected dimensions and we return the $k$ closest ones. 
This solution forces $O(2^D n^2)$ comparisons, which, for large sets of data, are clearly prohibitive.%

In this paper, we propose a scalable and efficient solution that aims at minimizing the number of distance computations to perform.
We leverage Locality Sensitive Hashing (LSH), a common technique in similarity search that groups together similar elements.
At a high level, we build an index of the time series where each multidimensional subsequence is mapped to a set of LSH hash values: by the properties of LSH this ensures that similar subsequences are more likely to hash to the same values across the dimensions spanned by the motif.
Thanks to this index we are able to prune a hefty amount of candidates from the search space.
We focus on comparisons based on the \textit{z-normalized Euclidean Distance}, however our approach can be generalized to other similarity measures.

One of the challenges of employing LSH is setting its parameters to ensure that good quality results are retrieved efficiently, since the precise setting of the parameters is data-dependent.
To overcome this challenge we design an index that automatically tunes its parameters depending on the data at hand while respecting user-specified limits on the memory to be used.

Our contributions are the following:
\begin{itemize}[leftmargin= 10pt]
    \item
        We design an approach for top-$k$ motif discovery in multidimensional time series,
        named \textsc{MOMENTI}.
        Our approach is based on Locality Sensitive Hashing and returns exact answers with a user-specified failure probability.
        Furthermore, we provide a theoretical analysis of the correctness of our approach and on its complexity in terms of distance computations carried out, along with several optimizations to speed up the execution.
    \item 
        We provide an open source implementation of our approach, that we use to carry out an extensive experimental evaluation.
        We show that our approach outperforms state-of-the-art baselines in terms of scalability,
        while providing high quality results.
\end{itemize}

\subsubsection*{Organization}
After reviewing the related work (\Cref{sec:relwork}) and introducing the background concepts of time series and motif discovery (\Cref{sec:ts}), we describe our approach (\Cref{sec:algo}). We then proceed in formalizing the details of our approach (\Cref{sec:complexity}) and we integrate optimizations to improve data adaptability and time complexity (\Cref{sec:opti}). Finally, we compare our proposal with other baselines and we extensively test our approach under different conditions to show its effectiveness (\Cref{sec:exp}).

\section{Related Work}
\label{sec:relwork}
A relatively large body of literature exists on time series and motif discovery, however the multidimensional case, despite aggregating a large interest, has seen only a fraction of the various works developed for the unidimensional case.

Comparisons between subsequences in time series are usually carried out by comparing their shapes rather than the raw values, in order to show invariance to noise and scale.
Several similarity measures have been used in time series processing, they can be split in two major subgroups: \textit{elastic} measures, that create a non linear one-to-many mapping between points of sequences, and \textit{lock-step} measures, where the mapping is one-to-one \cite{distances}. Between the most commonly used distances is \textit{Dynamic Time Warping} (DWT), an \textit{elastic} measure that allows the comparison of temporally misaligned sequences due to compression or stretching of shapes (i.e., warp in time).
\textit{Edit Distances}, are a family of distances that measure the number of \textit{edits} (e.g., substitutions, deletions, insertions) needed to obtain equal subsequences \cite{xiao2019edit}. Another common measure is the \textit{z-normalized Euclidean Distance}, a lock-step measure that z-normalizes the data before computing the Euclidean Distance, allowing variations in amplitude and mean values, so that the measured similarity is between the \textit{shapes} \cite{de2019implications}. 
This distance measure is, up to a constant factor, equivalent to the Pearson correlation coefficient \cite{berthold2016pearseucl}.

Many motif discovery techniques rely on symbolic abstraction of the raw data to facilitate matching of common patterns, besides smoothing out noise in the data. \textit{SAX} \cite{sax} found great success for its efficiency, requiring only a mean for the \textit{Piecewise Aggregate Approximation} and a table look-up for the symbol association.
It is less computationally complex than symbolization methods like \textit{ACA} \cite{sant2011symbolization}. Moreover, it is more general compared to methods like \textit{Persist} \cite{persist2005}, which require time series with a recognizable underlying structure \cite{sant2011symbolization}. We refer to the survey by \added{\citet{symbolizationreview}} for a complete overview on symbolization techniques.

Approximate algorithms for multidimensional motif discovery can be categorized in two major families: those that reduce in some way the time series into a unidimensional one and those that effectively work on multidimensional data.

Algorithms in the first category use techniques like \textit{Principal Component Analysis} (PCA) to generate a \textit{meta-}unidimensional time series that can be processed with the standard approaches developed for motif discovery in the univariate case. The work \replaced{in}{} \cite{tanaka2005discovery} employs \textit{Minimum Description Length} (MDL), to find the motifs. This approach is based on the strong assumption that all dimensions are relevant, as even a small number of noise dimensions leads to a meta-time series with little to no information, the  work of \cite{tanaka2005discovery} asserts how this algorithm \textit{can extract
a motif that can be recognized intuitively by human}, underlining how more work is needed when the structure of the time series is unknown. Moreover, the authors highlight the challenge of dynamically tuning the input parameters, since suboptimal sets can lead to poor outcomes in discovery.

The approaches that fall into the second category can be divided into two subfamilies: those who find motifs that span simultaneously in all dimensions \cite{berlin2012detecting} and those who find \textit{subdimensional motifs}. The second category is the one that allows the extraction of the most amount of information, since finding motif in all dimensions falls into a similar assumption of the algorithms that synthesize an univariate meta-time series, considering also irrelevant dimensions in the process.

For the task of subdimensional motif discovery, \replaced{\citet{4470297} introduced}{} the \textit{Random Projection} algorithm, which applies \textit{SAX} to independently symbolize each dimension of the time series, then a matrix of collisions between the subsequences is populated by iteratively random selecting a set of dimensions, creating words by concatenating the selected symbols and finding the matches. The algorithm has a linear running time in expectation, but it is greatly affected by the input threshold on the distance, a data dependent variable that is difficult to have an idea of without knowing in great details the data and the kind of motif searched.

The \textit{matrix profile} \cite{keoghMP} is commonly used to solve this problem exactly. It is a data structure that stores the distance between a subsequence and its nearest neighbor. The first motif can be identified by searching the minimal entry in the matrix, the second is the next minimum not overlapping with the first one, and so on. When discovery is limited to motifs with a certain dimensionality, the search can be restricted to only the $d$-th row of the matrix. To find the set of dimensions that span the motif, a variety of techniques can be used (e.g., finding the subset with minimal distance) with MDL being the one used by the state-of-the-art implementation \cite{Law2019}.

Very recently, a different variant of the problem has been introduced~\cite{DBLP:journals/pvldb/SchaferL24}:
the goal is to find a set of $k$ subsequences
of the input, spanning a limited number of dimensions and
minimizing the maximum pairwise distance of subsequences in the set.
 
Locality Sensitive Hashing (LSH) is a technique often employed in similarity search~\cite{lsh, 2020mining} which we will review in the next section.
Relevant for the scope of this paper is the family of hash functions for the Euclidean distance \cite{datar2004}.

LSH has already been used in the context of time series to discover motifs: \replaced{\citet{LSHearthquake}}{} employed LSH to derive fingerprints for earthquake waveforms.
In the one dimensional case \citet{attimo} provided an algorithm with guarantees on recall by employing the properties of LSH. Subsequences are matched using their \textit{fingerprints}, the distances are verified only on matching fingerprint pairs, allowing the algorithm to compute just a fraction of all the distance computations.

Using LSH requires setting up a number of parameters. The framework developed by PUFFIN \cite{puffinn} is capable of automatically tuning these parameters for $k$-nearest neighbors queries. \replaced{\citet{confsampling}}{} developed a technique to automatically find the best number of concatenations and repetitions to successfully answer a nearest neighbor query with a certain probability even when the probability of collision between pairs of points is unknown. We follow the \replaced{PUFFINN implementation~\cite{puffinn}}{}, expanding on the case where the pairs are actually ordered sets.

\begin{figure*}[t]
    \centering
    \includeinkscape[width=\linewidth]{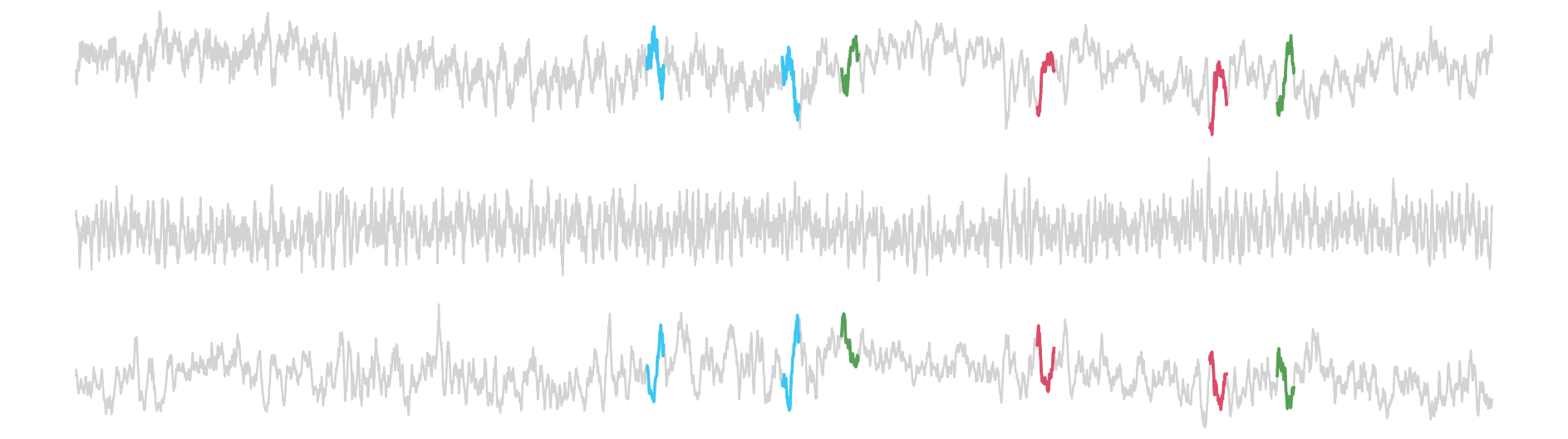}
    \caption{
        Multidimensional time series from an industrial evaporator \cite{DaISyEVAP}.
        The top-3 two dimensional motifs are highlighted. }
    \label{fig:multits}
    \Description[A time series with 3 dimensions, 3 motifs are highlighted in the first and third dimensions.]{The first and third dimension represent steam pressure and temperature, the second is an indicator that acts like noise in this context. The motifs represent some normal behaviors where one dimension increases while the other decreases.}
\end{figure*}
\section{Preliminaries}
\label{sec:ts}
\begin{table}[t]
    \centering
    \caption{Table of symbols.}
    \begin{tabular}{p{.2\columnwidth}p{.7\columnwidth}}
        \specialrule{1pt}{0pt}{0pt}
        $\Tmulti$ & Multidimensional time series \\
        $\mathbf{T}^w$ & Set of subsequences from $\mathbf{T}$ of length $w$\\
        $\subTa$ & Subsequence at index $a$\\
        $\subTaf$ & Dimension $f$ of subsequence $\subTa$\\
        $D$ & Dimensionality of the time series \\
        $d$ & Dimensionality of the motifs to discover \\
        $\dist{\cdot, \cdot}$ & Distance function \\
        $\distd{\cdot, \cdot}$ & $d$-dimensional distance function \\
        $\dims{\cdot, \cdot}$ & dimensions involved in the computation of the $d$-dimensional distance function \\
        $k$ & Number of motifs to find \\
        \hline
        $h_{i,j}\left(\subTaf \right)$ & hash value of length $i$ at repetition $j$
                                         for dimension $f$ of subsequence $\subTa$ \\
        $L$ & Number of repetitions for LSH \\
        $K$ & Number of concatenations for LSH \\
        \specialrule{1pt}{0pt}{0pt}
    \end{tabular}
    \label{tab:notation}
\end{table}
\subsection{Time Series and Motifs}
We will start our introduction from the special case of a single dimension time series.
\begin{definition}
        \textup{A} \textit{time series} $T\in \mathbb{R} ^n$ \textup{is an ordered sequence of real valued numbers} $T = [t_1,...\text{ }t_n]: t_u\in \mathbb{R} \text{, }\forall i\in[1,n]$ \textup{where $n$ is the length of the time series}.
\end{definition}
\begin{definition}
    \textup{A \textit{time series subsequence} $T_{u,w}\in \mathbb{R}^w$ is a subset of adjacent elements from T starting at position $u$ and length $w$, $T_{u,w}=[t_u,...\text{ }t_{u+w-1}]$.
    }
\end{definition}
\begin{definition}
    The \textit{z-normalized Euclidean Distance} between two time series subsequences $T_{u,w}$ and $T_{v,w}$ is defined as:
    \begin{equation}
    \label{eq:zned}
    \dist{T_{u,w},T_{v,w}} = 
    \sqrt{\sum_{i\in[w]}
        \left( \frac{T_{u}(i)-\mu(T_{u,w})}{\sigma(T_{u,w})}
        -
        \frac{T_{v}(i)-\mu(T_{v,w})}{\sigma(T_{v,w})} \right)^2
    }
    \end{equation}
    where $\mu$ and $\sigma$ are the means and the standard deviations of the subsequences, respectively, and $T_{u}(i)$ is the $i$-th value of subsequence $T_{u,w}$.
\end{definition}
\begin{definition}
    \textup{A \textit{time series motif} is the pair of subsequences $T_{u,w}$ and $T_{v,w}$: \begin{equation*}
        \dist{T_{u,w},T_{v,w}}\leq \dist{T_{a,w},T_{b,w}}
    \end{equation*} for all $a,b \in [1,...\text{ }n-w+1]$.
    }
\end{definition}
It comes that overlapping sequences are more likely to satisfy the condition above, these pairs are excluded in the problem formulation by ensuring no trivial matches.
\begin{definition}
    A pair of subsequences $T_{u,w}$ and $T_{v,w}$ is a \textit{trivial match} when: $v=u$ or 
        if, given an exclusion zone $\epsilon>0$, then $|u-v|\le\epsilon$. 
\end{definition}
Common values for the exclusion zone $\epsilon$ are fractions of $w$ to exclude matches in the neighboring area of a subsequence \cite{mp1,mp2}. We employ $\epsilon=w/2$.

Expanding upon the single dimensional case, the notions for the \textit{multidimensional} case are derived.
\begin{definition}
    \textup{A \textit{multidimensional time series} $\textbf{T}\in \mathbb{R}^{n\times D}$ is a tuple of $D$ time series $T^{(j)}\in \mathbb{R}^n : \textbf{T}=\left(T^{(1)},...\text{ }T^{(D)}\right)$ where $D$ is the dimensionality of $\textbf{T}$ and $n$ is its length.   
    }
\end{definition}
\begin{definition}
    \textup{A \textit{multidimensional subsequence} $\textbf{T}_{u,w}\in \mathbb{R}^{w\times D}$ is a tuple of unidimensional subsequences from $\textbf{T}$ starting at position $u$ and length $w$, $\textbf{T}_{u,w}=\left(T^{(1)}_{u,w},...\text{ }T^{(D)}_{u,w}\right)$.
    }
\end{definition}
When clear from context we will omit the subscript $w$ from the subsequence notation.

Given two multidimensional subsequences, we are interested in computing their distance.
However, as we mentioned oftentimes considering all dimensions in this computation is uninteresting at best, misleading at worst.
Following \cite{keoghMP} we therefore consider only a subset of $d$ dimensions
when computing the distance.

\begin{definition}
\label{def:distdims}
The $d$-dimensional distance between two subsequences $\subTa$ and $\subTb$
is
\[
\distd{\subTa, \subTb}
=
\min_{F \subseteq 2^{[D]}, |F| = d} \sum_{f \in F} \dist{\subTaf, \subTbf}
\]
\end{definition}
In other words, we select a subset of the dimensions of size $d$ such that the sum of the distances between the individual dimensions is minimized.
Similarly, with
\[
\dims{\subTa, \subTb}
=
\argmin_{F \subseteq 2^{[D]}, |F| = d} \sum_{f \in F} \dist{\subTaf, \subTbf}
\]
we denote the $d$ dimensions that define the distance between two subsequences.
Note that this set of dimensions is potentially different for different pairs of subsequences.
Furthermore, for a given pair of subsequences we name the distance of the dimensions maximally far apart among the ones belonging to $\dims{\subTa, \subTb}$:
\begin{equation}\label{eq:distdmax}
\distdmax{\subTa, \subTb} =
\max_{f \in \dims{\subTa, \subTb}}
\dist{
    \subTaf, \subTbf
}
\end{equation}

\begin{definition}
    \label{def:dist}
    A $d$-dimensional motif is the pair of subsequences $\subTa$, $\subTb$
    such that
    \[
        \distd{\subTa, \subTb}
        \le
        \distd{\subTmulti{u}, \subTmulti{v}}
        \quad
        \forall u, v \in [n-w+1]
    \]
\end{definition}

We are interested in finding the most similar subsequences in an unknown subspace.
\begin{definition}
\textup{Given a $D$-dimensional time series $\mathbf{T}$, motif length $w$, motif dimensionality $d$, and a distance function $dist$, the \textit{top-k multidimensional motifs} are the $k$ subsequences and their subspaces $F$ that minimize the $d$-dimensional distance with respect to all other subsequences of length $w$ in $\mathbf{T}$, ensuring no trivial matches between all possible pair of indices that are part of the motif pairs (i.e., no motif overlaps with another).}
\end{definition}

\begin{example}
    \Cref{fig:multits} reports an example of multidimensional motif discovery:  readings of different sensors from an industrial evaporator \cite{DaISyEVAP}. We highlight the top-3 motifs of length $w = 75$ and dimensionality $d=2$.
    Notice how the middle signal does not participate in the motifs.
\end{example}

\subsection{Locality Sensitive Hashing}

A powerful technique for approximate similarity search in high-dimensional space is Locality Sensitive Hashing (LSH for short), which we briefly introduce here.
Given that the subsequences of length $w$ of a time series can be seen as vectors in $\mathbb{R}^w$, this will prove a useful tool in this setting as well.
For an in-depth discussion of LSH, refer to \cite{lsh, wang2014hashingsimilaritysearchsurvey}.
We provide, as online supplemental material, a short interactive LSH primer\footnote{\url{https://www.dei.unipd.it/~ceccarello/MOMENTI-supplemental/}}.

Intuitively, LSH partitions a set of vectors randomly in such a way that close vectors are more likely to end in the same part than far away vectors.
To formalize this intuition, the definition below introduces a distance threshold $R$:
vectors that are closer to each other than $R$ are considered close, and vectors farther than
$cR$ are considered far away.

\begin{definition}[Locality Sensitive Hashing~\cite{lsh}]
\label{def:lsh}
Let $(\mathcal{X}, \operatorname{dist})$ be a metric space
and $\mathcal{H}$ be a family of functions $h: \mathcal{X} \to U$ for some set $U$. 
For a distance threshold $R$ and a constant $c>1$ the family $\mathcal{H}$ is called
$(R,cR,p_1,p_2)$-\textit{locality sensitive}
if $\forall x,y \in \mathcal{X}$ and
for $h$ sampled from $\mathcal{H}$:
\begin{equation}
\label{eq:lsh}
\begin{split}
    \text{if }
        \operatorname{dist}(x, y)\leq R 
        ~~\text{ then } ~~
        \Pr_{h\sim \mathcal{H}}\left[h\left(q\right)=h\left(p\right)\right]\geq p_1 \\
    \text{if }
        \operatorname{dist}(x, y)\geq cR 
        ~~\text{ then }~~
        \Pr_{h\sim H}\left[h\left(q\right)=h\left(p\right)\right]\leq p_2
\end{split}
\quad.
\end{equation}
\end{definition}
The event of two vectors having the same hash value is called a \emph{collision}.
A key quantity to assess the performance of LSH families is
\[
\rho=\frac{\log 1/p_1}{\log 1/p_2}
\]
A small $\rho$ value entails that the LSH family is effective at discerning close vectors from far-away ones.

For the common case of the Euclidean distance considered in this paper a widely used LSH family is that of \emph{Discretized Random Projections}~\cite{datar2004}.
For a vector $x\in \mathbb{R}^w$ and quantization parameter $r\in \mathbb{R}^+$ the hash function is
\begin{equation}
    h(x) = \left\lfloor \frac{a\cdot x+b}{r} \right\rfloor
    \label{eq:drp}
\end{equation}
where $a\in\mathbb{R}^w$ is a vector with random components following the $\mathcal{N}(0,1)$ Gaussian distribution, and $b\in\mathbb{R}$ is chosen uniformly at random in the interval $[0,r]$.

The probability that two vectors $x,y$ at Euclidean distance $R$ collide is:
\begin{equation}\label{eq:eucl-collision-probability}
    \Pr_{h\sim\mathcal{H}}\left[h(x) = h(y)\right] = 1-2\cdot norm\left(-\frac{r}{R}\right)-\frac{2}{\sqrt{2\pi}~r/R}\left(1-e^{-\left(\frac{r^2}{2R^2}\right)}\right)
\end{equation}
where $norm$ is the cumulative distribution function of a \textit{Standard normal distribution} \cite{datar2004}.
For this family of LSH functions we have $\rho = 1/c$ \cite{rhoboundlsh},
a fact that we will use in the analysis of the complexity of our algorithm.

As a notational shorthand, for two vectors $x$ and $y$ at distance $\dist{x, y}$ we define
\[
P(\dist{x, y}) = \Pr_{h\in \mathcal{H}}[h(x) = h(y)] ~.
\]

In order to \emph{amplify} the gap between the collision probability of close vectors
(at distance $\le R$) and far vectors (at distance $\ge cR$) a common strategy
is to create a \emph{composite} hash function by sampling $K$ hash functions and concatenating their outputs in a tuple of length $K$:
\[
    h'(x) = \langle
        h_1(x), h_2(x), \dots, h_K(x)
    \rangle
\]
The resulting LSH family is $(R, cR, p_1^K, p_2^K)$-locality sensitive.
Using larger values of $K$ lowers the probability for both close and far points to collide, with a more marked effect on the latter.
Repeating this process with $L$ independent composite hash functions implies that
points at distance smaller than $R$ collide in at least one of the repetitions with probability at least $1-\left(1- p_1^K\right)^L$.

Setting the parameters $K$ and $L$ requires the knowledge of the distance threshold
$R$, which in our setting is the distance of the top-$k$ motif.
Of course we do not know this distance beforehand, therefore in the following we describe
an algorithm that automatically tunes them based on the input.

\begin{algorithm}[t]
    \SetKwInOut{Input}{\raggedright{Input}}\SetKwInOut{Output}{Output}
    \Input{$D$-dimensional time series $\mathbf{T}$, subsequence length $w$, number of motifs $k$, $K$ max allowed number of concatenations, $L$ max number of repetitions, dimensionality of the motifs to find $d$, failure probability $\delta$}
    \Output{\{Set of the top-$k$ motifs\}, with probability $1-\delta$}
    \tcp{Initialization}
    \For{$\mathbf{T}_a \in \mathbf{T}^w,$ \added{$f \leftarrow 1 \text{ to } D, j\leftarrow 1 \text{ to } L$}}{
            \label{ln:hash-computation}
        \added{compute $h_{K,j}(\subTaf)$}\;
    }
    \BlankLine
    
    \texttt{TOP = PriorityQueue()}\label{ln:priority-queue-init}\;
    \For{$i\leftarrow K \text{ to } 1$}{ \label{ln:prefix-cycle}
        \For{$j\leftarrow 1 \text{ to } L$}{
            \Let{$E$}{$\emptyset$}\label{ln:define-edges}\;
            \For{\added{$f \leftarrow 1 \text{ to } D$}}{%
                \For{$(\subTa, \subTb) \in \Tmulti^w \times \Tmulti^w : h^f_{i,j}(\subTaf) = h^f_{i,j}(\subTbf)$}{
                    \label{ln:prefix-eq}
                    \uIf{$(\subTa, \subTb) \notin E$}{
                        \Let{$E$}{$E \cup \{(a,b)\}$}\;
                        \Let{$W(a,b)$}{1}\;
                    }
                    \Else {
                        \Let{$W(a,b)$}{$W(a,b)+1$}\label{ln:increment-weight}\;
                    }
                } 
            }
            \For{$(a, b) \in E$}{
                \label{ln:second-start}
                \If{$W(a, b) \ge d$}{
                    \label{ln:weigth-constr}
                    \texttt{TOP.insert}$\left((\subTa,\subTb)\right)$\;
                    \If{$|$\texttt{TOP}$|$>k}{
                        \texttt{TOP.pop()} \label{ln:second-end}
                    }
                }
            }
            \If{$|$\texttt{TOP}$|=k$ $\wedge$ \texttt{STOP}$\left(\text{\texttt{TOP.max()}},i,j,\delta\right)$}
            {\Return \texttt{TOP}}
        }
    }
    \Return true top-$k$ by computing all pairs\;
\caption{MOMENTI}
\label{alg:emitaggr}
\end{algorithm}
\begin{algorithm}[t]
    \caption{Stopping condition\label{alg:stopping-condition}}
    \SetKwProg{Fn}{Function}{ is}{end}
    \Fn{\texttt{STOP}$((\subTa, \subTb), i, j, \delta)$}{
        \Let{$p$}{
            $P\left(\distdmax{\subTa, \subTb}\right)^{d}$
        }\label{ln:collision-probability}\;
        \lIf{$i=K$}{
            \Return 
                $\left(1-p^i\right)^j \le \delta$ \label{ln:case-k}
        }\lElse {
            \Return 
                $\left(1-p^i\right)^j \cdot \left(1-p^{i+1}\right)^{L-j} \le \delta$
                \label{ln:case-shorter-than-k}
        }
    }
\end{algorithm}

\section{Algorithm}
\label{sec:algo}

We now describe our algorithm, named \textsc{MOMENTI}
(for \emph{MOtifs in MultidimEnsioNal TImeseries}),
to find the top-$k$ motifs in a multidimensional time series. Our algorithm has a user-defined error probability $\delta$.
At a high level, our algorithm is comprised of two main phases.
First, it builds a LSH-based index of the time series subsequences, where the dimensions of all subsequences are hashed independently.
Then the index is traversed to discover candidate motif pairs, until a data-dependent stopping condition is met.
The pseudocode of our algorithm is presented in Algorithm~\ref{alg:emitaggr}.

\paragraph{Index construction}
As we have seen in the previous section, two critical parameters in an LSH setup are the number of concatenations $K$ and the number of repetitions $L$.
Furthermore, the hash function of Equation~\eqref{eq:drp} requires a quantization parameter $r$ to be set as well. We shall see how to automatically set $r$ later in Section~\ref{sec:opti}.

For a time series $\Tmulti$, setting $K$ and $L$ so to minimize the number of distance computations requires the knowledge of the motif distance \emph{before} we construct the index.
To work around the fact that we, of course, do not know this distance beforehand, our index instead sets the \emph{maximum} $K$ and $L$ values to be used in the second phase.

To construct the index of the subsequence of length $w$ of a multivariate time series $\Tmulti$
we first sample multiple independent composite hash functions of length $K$: one for each dimension $f\in[D]$ and repetition $j\in[L]$. As described in the previous section, this is achieved by simply sampling a random vector $a$ and a random value $b$ for each function.
We denote the composite hash function at repetition $j\in[L]$ for dimension $f\in[D]$ with
$h_{K,j}^f$.

Then, for each subsequence $\subTa$, we compute multiple independent hash values:
for each dimension $f\in[D]$ and for each repetition $j\in[L]$, we compute
the \emph{composite} hash value of length $K$ of the vector $\subTaf$ using the corresponding hash function, that is we compute $h_{K,j}^f(\subTaf)$.

Note that a composite hash value can be seen as a string of $K$ integer values.
A fundamental operation in the next phase will be retrieving, for a given dimension $f\in[D]$
and repetition $j\in[L]$, all the subsequences whose hash share the prefix of a given length.
\added{To efficiently support this operation, we construct an index on hash values consisting of a family of ordered vectors. Specifically, for each dimension $f\in[D]$ and each repetition $j\in[L]$, each vector indexes the string of $K$ hash values of all subsequences. By doing so, hash values with the same prefix appear in contiguous ranges of the array, providing higher locality of reference.}
In the following, for $0 <i \le K$ we denote with $h_{i, j}^f(\subTaf)$ the prefix
of length $i$ of the hash value for the subsequence dimension $\subTaf$ in repetition $j$.

\paragraph{Index traversal}

A key property of the index introduced in the previous paragraph is that very similar subsequences are likely to share long prefixes of their hashes.
At the same time, due to the probabilistic nature of LSH it is not certain that similar
subsequences share a long prefix in the first repetition.

We first give the intuition behind the algorithm, and then give all the details.
Starting from the longest possible hash prefix, $K$, the $L$ repetitions are considered,
focusing on each one on all pairs of subsequences
sharing the same hash prefix: for each such pair we can compute the distance, which is used to rank candidate motif pairs in a priority queue.
Furthermore, from the distance we can derive the collision probability by means of Equation~\eqref{eq:eucl-collision-probability}.
We use this probability to define a stopping condition that allows to rule out the event that the true motif has \emph{not} been seen in the repetitions considered so far.
If the probability of this negative event is less than a user defined threshold $\delta$ and the priority queue contains at least $k$ elements then
the algorithm stops, returning the top-$k$ pairs in the priority queue.
Otherwise the next iteration is considered, with a caveat: if the last of the $L$ repetitions is reached with the stopping condition not satisfied, then it means that hashes of length $K$ are too selective for the dataset at hand. Therefore, the process is restarted considering prefixes of length $K-1$. This procedure continues, potentially considering shorter and shorter prefixes, until the stopping condition is met.

The above high level intuition is complicated by the fact that we have to deal with multidimensional subsequences whose dimensions are hashed independently. In the following we thus detail the algorithm with reference to the pseudocode in Algorithm~\ref{alg:emitaggr}.

The algorithm proceeds in rounds from $K$ to $1$: in round $i$ hash prefixes of
length $i$ are considered.
In each round $i$ then all repetitions $1\le j \le L$ are considered,
and each such iteration is comprised of three steps:
counting the number of collisions, computing distances, and checking the stopping condition.

The first step (line~\ref{ln:define-edges} to~\ref{ln:increment-weight}) counts for each pair of subsequences the number of dimensions in which they share a prefix of length $i$.
Note that retrieving the pairs of subsequences colliding in each dimension (line~\ref{ln:prefix-eq}) is done efficiently by leveraging the fact that hash values are stored in \replaced{lexicographically sorted arrays}{}.
By the end of the first step the algorithm has built a set $E$ of pairs of subsequences
that collided in at least one dimension, and for each pair $(a,b) \in E$ the function
$W(a,b)$ reports the number of dimensions on which $\subTa$ and $\subTb$ collided.

In the second step the focus (lines from~\ref{ln:second-start} to~\ref{ln:second-end})
is on the pairs that share prefixes in at least $d$ dimensions out of $D$, where $d$ is the target number of dimensions spanned by the motifs.
The intuition is that a motif pair $(\subTa, \subTb)$, being similar in at least $d$ dimensions, will have $W(a, b) \ge d$.
For each such pair, the algorithm computes the distance and updates the priority queue of candidates, keeping only the top-$k$ in memory.

Finally, if the priority queue contains $k$ candidates then the algorithm checks the stopping condition (Algorithm~\ref{alg:stopping-condition}), which considers the pair at maximum
distance in the priority queue. In particular, for this pair $(\subTa, \subTb)$ the stopping condition focuses on $\distdmax{\subTa, \subTb}$ as defined in Equation~\eqref{eq:distdmax}.
This distance is used to compute an upper bound $p$ to the probability of the two subsequences colliding in $d$ dimensions at the same time independently (line~\ref{ln:collision-probability}).
This probability $p$ is then used to compute the probability that a pair of subsequences with a smaller pairwise distance was missed by the algorithm in all the previous iterations.
If this probability is smaller than the user-defined $\delta$, then the algorithm can successfully terminate.
In particular, line~\ref{ln:case-k} is evaluated when the full hashes are being considered, and only the first $j$ repetition have been executed.
Line~\ref{ln:case-shorter-than-k} is executed when prefixes shorter than $K$ are under consideration, and takes into account the fact that the algorithm executed $L-j$ iterations with prefixes of length $i+1$.

\added{Note that the iterative process of \crefrange{ln:prefix-cycle}{ln:weigth-constr} inherently exhibits \textit{anytime} properties, as the priority queue only allows the insertion of better solutions. The user can stop at any time the discovery to retrieve the candidate motifs with their error probabilities. }

 \paragraph{Example 4.1}
\added{
 Consider the time series shown in \Cref{fig:multits}, let $\mathbf{T}_a, \mathbf{T}_b$ be the names of the subsequences belonging to the first motif (in red) and $\mathbf{T}_c, \mathbf{T}_d$ to the second (in blue), \texttt{TOP} be the queue where we store the top motif with $k=1$, $K=4$, $L=2$ and $\delta=0.1$. Let us assume that there exist no other collisions other than the two motifs. During the iteration at $i=4$, $j=1$ we will scan the hashes for each dimension and update our weights (lines 7-13), we find $W(\mathbf{T}_a, \mathbf{T}_b)=W(\mathbf{T}_c, \mathbf{T}_d)=2$. The first couple will be stored in the queue since it has a lower distance (lines 14-18).
 At this point (line 19), \texttt{STOP}$\left(\distd{\mathbf{T}_a, \mathbf{T}_b}, 4, 1\right) \geq \delta$ so the condition is not satisfied. This happens again for $j=2$, when we find $(\mathbf{T}_c, \mathbf{T}_d)$ we evaluate the distance again because we have no memory of the past, and discard it for its distance is greater than the couple stored in \texttt{TOP}. Since the hashes of length $i=4$ were unable to satisfy the stopping condition, we move to prefix length $i=3$.  This time, for both $j=1$ and $j=2$, we find $(\mathbf{T}_a, \mathbf{T}_b)$ and $(\mathbf{T}_c, \mathbf{T}_d)$ with weight $2$, we will compare them again even though they already have been evaluated during the previous step, we will later implement a way to avoid this operation and save resources.
 At the end of repetition $j=2$ we find that the bound on the failure probability of the pair stored in the queue is satisfied and the algorithm stops.
}

\section{Analysis}
\label{sec:complexity}
In this section, we derive the probabilistic guarantees for our algorithm as well as analyze its complexity.

\begin{lemma}\label{lem:collision-probability}
Given a pair of subsequences $\subTa, \subTb$ and parameter $d_m$,
consider iteration
$i$ of the outer loop of \Cref{alg:emitaggr}.
Then, we have $W(a, b) \ge d$ with probability at least
$
    P\left(\distdmax{\subTa, \subTb}\right)^{i \cdot d}
$
where $\distdmax{\cdot,\cdot}$ is defined as in Equation~\eqref{eq:distdmax}.
\end{lemma}
\begin{proof}
    To have $W(a,b) \ge d$ we need the two subsequences to collide in at least $d$ dimensions.
    \newline
    Consider the distance $\distdmax{\subTa, \subTb}$ and recall that, by definition, all dimensions
    $F=\dims{\subTa, \subTb}$ will be at a closer distance.

    At iteration $i$ of the outer loop of \Cref{alg:emitaggr} we have that the two subsequences collide in all dimensions $F$ at the same time with probability at least
    $
    \prod_{f\in F} P\left(\dist {\subTaf, \subTbf}\right)^i
    $
    given that $\bar{f}$ is the dimension of maximum distance out of the $d$ ones in $F$, each
    factor of the above product is lower bounded by 
    $P\left(\subTsingle{a}{\bar{f}},\subTsingle{b}{\bar{f}}\right)$,
    hence the statement follows.
\end{proof}

\subsection{Correctness}

To prove correctness we need to show that motif pairs are considered at least once before the algorithm terminates.

The following lemma bounds the probability that a given pair of subsequences at indices $a$ and $b$
never has a weight $W(a,b) \ge d$. In other words, the following lemma bounds the probability that a pair is never considered for inclusion in the \texttt{TOP} priority queue.

\begin{lemma}
    \label{lemma:stop}
    Let $(\subTa, \subTb)$ be a pair of subsequences and let
    \[
    \begin{aligned}
    p = P\left(\distdmax{\subTa,\subTb}\right)^d
    \end{aligned}
    \]
    Consider iteration $i$ out of $K$ of the outer loop
    and iteration $j$ out of $L$ of the inner loop of Algorithm~\ref{alg:emitaggr}.
    The probability that $W(a,b) < d$ in all previous iterations is
    upper bounded by
    \begin{equation}
    \label{eq:probnotcolliding}
      \left\{
        \begin{aligned}
        &(1-p^i)^j \quad &\textrm{ if }~~ i=K \\
        &(1-p^i)^j \cdot (1-p^{i+1})^{L-j} \quad&\textrm{ otherwise}
        \end{aligned}
      \right.
    \end{equation}
\end{lemma}
\begin{proof}Sketch of the proof:
from Lemma~\ref{lem:collision-probability} we know the probability of the pair to collide in a single repetition with prefixes of length $i$. For the given pair the failure event is to have a weight $W(a,b)<d$. Therefore, the probability of never colliding over $j$ independent repetitions is $(1-p^i)^j$.
In the case where $i=K$, the statement follows.
When $i <K$, we consider that the first $j$ iterations performed with prefix $i$ fail independently and that the $L-j$ iterations previously run with prefix $i+1$ failed as well.
\end{proof}

From this we now derive two lemmas on the success probability of the discovery.
In \Cref{lemma:topk} we allow each motif to fail independently, in \Cref{lemma:topkdip} we require that all the returned motifs are correct within a probability.
\begin{lemma}
    \label{lemma:topk}
    \Cref{alg:emitaggr} finds the true top-$k$ motifs, each with probability $\geq 1-\delta$.
\end{lemma}
\begin{proof}
    Let the stopping condition be met at iteration $i'$ and concatenation $j'$, $m_1,...m_k$ be the set of motifs returned by the algorithm, sorted by increasing distances. We have that the failure probability of $m_h,\text{ } h\in[1,k]$ is upper bounded by the failure probability $m_k$ for the monotonicity of the collision probability. This failure probability is given by \Cref{lemma:stop}, by construction the stopping condition ensures that this probability is $\leq \delta$. This is valid for all returned pairs independently. 

    If the stopping condition is never met the algorithm reaches line 20, 
    where all pairs of subsequences are considered, thus returning the correct motifs with probability 1.
\end{proof}
Given this lemma it is easy to derive the expected recall of the algorithm, that corresponds to $1-\delta$.
\begin{lemma}
\label{lemma:topkdip}
    When called with failure probability $\delta'=\delta / k$, \Cref{alg:emitaggr} finds the true top-$k$ motifs with probability $\geq 1-\delta$.
\end{lemma}
\begin{proof}
    From \Cref{lemma:topk} we have that each pair fails independently with probability $\leq \delta'= \delta /k$.
    By applying a union bound on the $k$ pairs we obtain the statement.
\end{proof}
\subsection{Number of Distance Computations}

First, we introduce the concept of \emph{contrast}, which will be useful in capturing the difficulty of a dataset and in relating it to the complexity of our algorithm.

\begin{definition}
\label{sec:contrast}
For a $D$-dimensional time series $\Tmulti$ of length $n$, and for parameters $k$ and $d$ let
$(\subTmulti{a_k}, \subTmulti{b_k})$
and
$(\subTmulti{a_n}, \subTmulti{b_n})$
be the $k$-th and $n$-th motifs, respectively.
We define
\[
\operatorname{contrast}_{d,k|n}(\Tmulti) = \frac{
  \distdmax{\subTmulti{a_n}, \subTmulti{b_n}}
}{
  \distdmax{\subTmulti{a_k}, \subTmulti{b_k}}
}
\]
\end{definition}

The following theorem ensures that our algorithm computes, in expectation, a subquadratic number of distances
in expectation, assuming we give to the algorithm enough memory.
The complexity is parameterized by the contrast of the motifs in the multivariate time series:
a large contrast implies a smaller complexity.

\begin{theorem}
On a $D$-dimensional time series $\Tmulti$ of length $n$,
with parameters $k\ge 1$, $d\in[1, D]$, $\delta\in (0,1)$, 
Algorithm~\ref{alg:emitaggr} computes
\[
  O\left(
    n^{1 + 1/c} \log\frac{1}{\delta} + Lk
  \right)
\]
distances in expectation, where $c=\operatorname{contrast}_{d,k|n}(\Tmulti)$,
for $L \in \Omega\left(n^{1/c}\right)$.
\end{theorem}
\begin{proof}
\newcommand{\distcomps}[1]{\operatorname{ED}\left(#1\right)}
Let $\mathbb{T}$ be the set of all the subsequence pairs of $\Tmulti$.
For any pair of subsequence indices $(a, b)$, let
$
p_{ab} = P\left(\distdmax{\subTa, \subTb}\right)^d
$.

Consider 
$(\subTmulti{a_k}, \subTmulti{b_k})$
and
$(\subTmulti{a_n}, \subTmulti{b_n})$,
the $k$-th and $n$-th $d$-dimensional motifs of $\Tmulti$, and let
$p_1=p_{a_k b_k}$ and $p_2=p_{a_n b_n}$.
Furthermore, let
$\rho = \frac{\log 1/p_1}{\log 1/p_2}$.

For notational convenience, define the following operator that gives the expected number of
distance computations carried out in $j$ independent repetitions with hash prefixes
of length $i$:
\begin{equation*}
\distcomps{i, j}
=
j\cdot \Big(
\sum_{(\subTa, \subTb) \in \mathbb{T}} p_{ab}^{i}
\Big)
\end{equation*}

Recall that, for hash prefixes of length $i$, examining
$\frac{\log 1/\delta}{p_1^i}$ independent repetitions ensures that
the top $k$ motif pairs are seen at least once with probability at least $1-\delta$.
Using the notation shorthand defined above, the number of expected distance computations
in this case is
$
\distcomps{i, \frac{\log 1/\delta}{p_1^i}}
$.

Now define $\mathbb{T}_{>k}$ and $\mathbb{T}_{>n}$ as the sets of subsequence pairs
that are farther than the $k$-th and $n$-th motifs, respectively.
Then we have
\begin{equation}
\label{eq:bounding_T>n}
\begin{aligned}
\distcomps{i, \frac{\log 1/\delta}{p_1^i}}
\le&
\frac{\log1/\delta}{p_1^i}
\cdot
\left(
n +
\sum_{(\subTa, \subTb) \in \mathbb{T}_{>n}} p_{ab}^{i}
\right)  \\
\le&
\frac{\log1/\delta}{p_1^i}
\cdot
\left(
n + {n \choose 2} p_2^i
\right)
\end{aligned}
\end{equation}
where the inequality follows from the definition of $\mathbb{T}_{>n}$.

Defining
$i^* = \frac{\log n}{\log\frac{1}{p_2}}$
we have
${n \choose 2} p_2^{i^*} \le n$
and
$\frac{1}{p_1^{i^*}} =n^{\frac{\log 1/p_1}{\log 1/p_2}} = n^\rho$.

Therefore there is a prefix length $i^*$ for which the number of expected distance computations
to see the $k$-th motif colliding at least once is
\begin{equation}\label{eq:istar-distcomps}
\distcomps{i^*, \frac{\log 1/\delta}{p_1^{i^*}}}
=
O\left( n^{1+\rho}\log\frac{1}{\delta} \right)
=
O\left( n^{1+\frac{1}{c}}\log\frac{1}{\delta} \right)
\end{equation}
where $c = \operatorname{contrast}_{d,k|n}(\mathbf{T})$ and the equality follows
from the definition of $\rho$ for the LSH family we employ in our algorithm.

Now, let $i' \ge i^*$ be the largest prefix $i$ such that the stopping condition holds.
With probability $1-\delta$ the algorithm stops at prefix $i'$.
Conditioned on this event, the number of distance computations carried out at level $i'$ is
\begin{equation}\label{eq:iprime-distcomps}
\distcomps{i', \frac{\log 1/\delta}{p_1^{i'}}}
+
\distcomps{i'+1, L-\frac{\log 1/\delta}{p_1^{i'}}}
\end{equation}
Where the second term accounts for the repetitions considered in iteration $i'+1$ of the outer loop of the algorithm.
First we bound the first term of the addition.
Similarly to before, define $\mathbb{T}_{>k}$ as the set of subsequence pairs that are farther away than the $k$-th motif.
\begin{equation*}
\begin{aligned}
\distcomps{i', \frac{\log 1/\delta}{p_1^{i'}}}
&\le
    Lk
    +
    \frac{\log 1/\delta}{p_1^{i'}}\left(
    \sum_{(\subTa, \subTb)\in\mathbb{T}_{>k}} p_{ab}^{i'}
    \right)\\
&\stackrel{(a)}{\le}
    Lk
    +
    \frac{\log 1/\delta}{p_1^{i*}}\left(
    \sum_{(\subTa, \subTb)\in\mathbb{T}_{>k}} p_{ab}^{i^*}
    \right)\\
&\le
    Lk + \distcomps{i^*, \frac{\log 1/\delta}{p_1^{i^*}}}
    \stackrel{(b)}{=}
    O\left(
        Lk + n^{1+\frac{1}{c}} \log \frac{1}{\delta}
    \right)
\end{aligned}
\end{equation*}
where (a) follows from the fact that $p_1 \ge p_{ab}$
and (b) follows from Equation~\eqref{eq:istar-distcomps}.
The theorem follows by observing that the second term of Equation~\eqref{eq:iprime-distcomps} is a factor $1/p_1 = O(1)$ larger than the first term.
\end{proof}

\subsection{Index construction and size}
We now consider the contribution to the running time of \Cref{alg:emitaggr} given by the hash index and derive its space complexity.
\begin{lemma}
\label{lemma:line3}
    The hash construction at line~\ref{ln:hash-computation} of \Cref{alg:emitaggr} takes time $O(D\cdot K \cdot L \cdot n\log n)$.
\end{lemma}
\begin{proof}
    For each multidimensional subsequence we have to evaluate $D\cdot K \cdot L$ hashes. For a fixed hash function we can compute all the dot products in $O\left(n\log n\right)$ time, using the cyclical convolution theorem. The result follows.
\end{proof}

Given that for each of the $n$ subsequences we have to store, in each of the $L$ repetitions, a total of $D$ hashes of length $K$ we have the following.
\begin{theorem}
    \Cref{alg:emitaggr} has space complexity proportional to $O\left(K\cdot L \cdot D\cdot n\right)$.
\end{theorem}

\section{Implementation details}
\label{sec:opti}
In this section, we describe the key aspects of the implementation that impact the algorithm's running time and its motif discovery capabilities. Namely, the number of hash evaluations and comparisons and the challenge of not knowing the dimensionality of the motifs to discover.
\subsection{Index building}
\subsubsection*{Setting the quantization width r}
\label{sec:r_auto}
Recall from Equation~\ref{eq:drp} the formulation of the hash function we use.
Even though the choice of the parameter $r$ does not compromise the correctness of the algorithm, it is fundamental since it influences the performance.
A value that is too high will group many subsequences together at long indices, mitigating the filtering effect of LSH, a value that is too low will separate even the smallest perturbation between subsequences, forcing the algorithm to visit shorter prefixes.
To deal with this parameter automatically, the algorithm adopts an estimation-based heuristic.

First we sample a number of random vectors from $\sim \mathcal{N}(0,1)$, computing their dot product with a random sample of subsequences from the time series.
This produces an estimate of the distribution of the values that are discretized by Equation~\ref{eq:drp}.

Then, we discretize this empirical distribution into $256$ equal-width buckets.
The width of these buckets will be used as the parameter $r$ in the hash function.
The rationale is that by doing so we will be able to represent each hash value with a single byte.

Note that this heuristic can be implemented efficiently by leveraging on the \emph{cyclical convolution theorem}, the same method used by MASS for the \textit{Distance Profile} \cite{zhong2024mass}.
We can obtain, with one convolution between one of the random vectors and the sample of the time series, all the dot products for that vector.

\subsubsection*{Tensoring}
Tensoring is a technique to reduce the number of evaluations for the hash functions \cite{christianitensoring}.
Let $\mathcal{H}$ be a LSH family and $K,L\geq 1$ integers. 
For $m=\sqrt{L}$,
we define $\mathcal{H}_l$ as a set of $m$ hash functions sampled from $\mathcal{H}^{\frac{K}{2}}$ and similarly for $\mathcal{H}_r$. Then $\left(h_a,h_b\right) \in \mathcal{H}_l\times \mathcal{H}_r, 1\leq a,b\leq m$ provides $m^2$ repetitions with $Km$ evaluations.
Furthermore, let us define for $1\le j \leq L$:
\begin{equation}
h_{K,j}=\left(h_{\frac{K}{2},l},h_{\frac{K}{2},r}\right) \in \mathcal{H}_{l}\times \mathcal{H}_{r} \text{ where } 
\begin{cases}
    l = j  \div \sqrt{L}\\
    r = j \text{ }\bmod \text{ } \sqrt{L}
\end{cases}.
\end{equation}
\newline
The resulting hash is obtained from interleaving values from the selected left and right hash.
\added{
This reduces the hash evaluations from $K\cdot L$ to
$K\cdot \sqrt{L}$ at the cost of losing independence
between repetitions.
To accommodate this change we give this alternative formulation of Lemma~\ref{lemma:stop}, omitting the proof for the sake of space.
\begin{lemma}\label{lemma:stop-tensoring}
Under the same conditions of Lemma~\ref{lemma:stop}, using the tensoring
approach, the probability that $W(a, b)<d$ in all
previous iterations is upper bounded by
\[
\begin{cases}
P_t(i/2, j) \quad&\textrm{if }~ i=K\\
P_t\left(i/2, j\right) \cdot P_t\left((i+1)/2, L-j\right)\quad&\textrm{otherwise}
\end{cases}
\]
where $P_t(i', j')=(1-p^{\lceil i'\rceil})^{j'\div \sqrt{L}}\cdot(1-p^{\lfloor i'\rfloor})^{j'\bmod\sqrt{L}}$.
\end{lemma}
Using Lemma~\ref{lemma:stop-tensoring} in the stopping condition of Algorithm~\ref{alg:emitaggr} maintains the correctness guarantees of Lemmas~\ref{lemma:topk} and~\ref{lemma:topkdip}.
}

\begin{corollary}[lemma:line3]
    With the tensoring approach the index construction takes time proportional to $O(D\cdot K\cdot \sqrt{L} \cdot n\text{ log }n)$.
\end{corollary}

\subsection{Index traversal}
\subsubsection*{Comparisons on the fly}
In \Cref{alg:emitaggr}, maintaining the structure of line~\ref{ln:prefix-cycle} at runtime is very expensive, 
since it would require quadratic space to store the weights for, potentially, each pair of subsequences.
To deal with this problem, while maintaining the same theoretical approach, we just scan the index over each dimension and every time a unidimensional collision is seen we immediately compute $W(a,b)$, if greater or equal than the searched motif dimensionality $d$, we perform the insertion of the pair in the priority queue. This approach does not require any additional space at the cost of a slightly higher number of hash comparisons, since we may evaluate a colliding pair up to $D$ times.

\subsubsection*{Duplicate collisions within the same repetition}
Our approach dynamically chooses the length of the composite hashes by progressively iterating on the prefixes of the full ones, this results in looking into prefixes of decreasing size. Consequently, collisions at level $i$ are a \textit{superset} of the collisions at level $i+1$, in order to avoid unnecessary distance computations, the algorithm will check if the colliding hashes appear at the level $i+1$, skipping the pair in the positive case.
\subsection{Finding motifs of multiple dimensionalities}
\label{sec:multisub}
One of the critical points up to now is the fact that at input we require the number of dimensions $d$ that span the motif. In many real scenarios, it is possible to know only approximately the expected dimensionality of the pattern in a certain domain.

Our method can be easily extended to discover motifs spanning different dimensions,
allowing the user to specify a range of dimensionalities $d_{low}, d_{high}$ of the motifs they want to discover.
The data structure at line~\ref{ln:priority-queue-init} of \Cref{alg:emitaggr} becomes a set of $d_{high}-d_{low}+1$ priority queues, which are independently updated during the discovery process.
This allows the algorithm to reuse information from a single distance computation across the different requested dimensionalities. Since evaluating $\distd{\subTa,\subTb}$ for a pair of subsequences requires computing $\dist{\subTaf,\subTbf} \quad\forall f\in[D]$, we can efficiently derive $\distd{\subTa,\subTb} \quad \forall d\in[1,D]$.
The condition of line~\ref{ln:weigth-constr} is initially based on $d_{low}$ and gradually increases every time the lowest dimensionality has its motif confirmed.
\added{This approach allows to return solutions for each motif dimensionality as soon as their error probabilities satisfy the required quality, inheriting the anytime property.
}%

\begin{table}[t]
\caption{\textbf{Information about the evaluation datasets.}}
\centering
\begin{tabular}{@{}lrrrrr@{}}
\toprule
dataset          & n      & D & window & $d$ &$c_{d,1|n}$ \\ \midrule
potentials       & 2 500   & 8 &   50     &  8 & 6.11\\
evaporator       & 7 000   & 5 &   75    &   2 & 2.60\\
RUTH             & 14 859   & 32 &   500    &   4 & 3.24  \\
weather          & 100 057 & 8 &   5000     &  2 &   2.43  \\
whales           & 450 001 & 32 & 300 & 6 & 1.22*\\
quake            & 6 440 998 & 32 & 100 & 4 & 1.65* \\
electrical\_load & 6 960 008 & 10  &  1000 & 5 & 1.91*   \\ 
LTMM             & 25 132 289 & 6 & 200 & 3 & 5.31* \\
\bottomrule
\multicolumn{6}{r}{{\small * obtained through random sampling.}}
\end{tabular}
\label{tab:dataset}
\end{table}

\section{Experiments}
\label{sec:exp}
This section aims at answering the following questions:
\begin{itemize}[leftmargin= 10pt]
    \item How does our algorithm compare with the state-of-the-art?
    \item \added{How do parameters $K$, $L$, $r$ and $\delta$ influence the performance?}
    \item How does the algorithm scale with respect to the input size?
    \item Is the algorithm able to find motifs in high dimensional noisy time series?
\end{itemize}
\subsubsection*{Baselines}
Our algorithm is compared to \textsc{Mstump} \cite{Law2019}, the state-of-the-art implementation for the multidimensional matrix profile.
\added{
We stress that \textsc{Mstump} is an \emph{exact} algorithm that computes
more information than just the motifs: for instance it can be used to detect discords. Our aim is to investigate the gains that can be attained when only motifs are sought, and when a small failure probability is accepted.
}
Alongside, we consider the Extended Motif Discovery (\textsc{EMD}) algorithm \cite{tanaka2005discovery}, reimplemented in \cite{emdsilva}, \added{and the Random Projection (\textsc{RP}) algorithm \cite{4470297}, reimplemented by us, for comparisons with approximate approaches. We refer to our algorithm as \textsc{MOMENTI}, which stands for MOtifs in MultidimEnsioNal TImeseries.}
All algorithms were subject to a global timeout of 4 hours per execution for datasets under the million points, and 24 hours for larger datasets.
\subsubsection*{Experimental Setup}
The evaluation was carried out on a machine with a 8 core Intel Xeon W-2245 @ 3.90 GHz equipped with 128 GB of memory.

\subsubsection*{Datasets}
All the experiments are run on the following real datasets from different domains, whose details can be found in \Cref{tab:dataset}.
\begin{itemize}[leftmargin = 10pt]
    \item \textsc{Potentials} is a record of skin potentials %
    \cite{DaISyPREG};
    \item \textsc{Evaporator} is data from a four-stage evaporator to reduce the water content from products \cite{DaISyEVAP};
    \item \textsc{Ruth}, the Mel-spectrogram of the song \textit{Running Up That Hill} by Kate Bush, extracted with the following parameters: 32-Mel scale filters, 46 milliseconds short time Fourier transform window and 23 milliseconds hop;
    \item \textsc{Weather} is representing the hourly climate data near \replaced{Monash University, in Australia}{}, for about $10$ years \cite{godahewa_2021_5184708};
    \item \textsc{Whales} is obtained from data of an underwater passive acoustic network and is a 10 minute recording of humpback whales vocalizations \cite{NOAA_PIPAN_2021};
    \item \textsc{Quake} is a waveform from the Observatories and Research Facilities for European Seismology (ORFEUS) during the 2014 Aegean Sea earthquake \cite{quake};
    \item \textsc{El\_load} includes cleaned electrical consumption data in Watts for 20 households at aggregate and appliance level \cite{PMID:28055033};
    \item \textsc{LTMM} contains 3-day 3D accelerometer recordings of elder community residents, used to study gait, stability, and fall risk \cite{falldataset, fall}.
\end{itemize}

The choice of the reference motif dimensionality $d$ for each dataset is guided by the additional information available. Specifically, for \textsc{Potentials}, \textsc{Evaporator} and \textsc{Weather} we use metadata information for the expected dimensionality of the pattern. For \textsc{Ruth} and \textsc{Whales}, we follow previous work on audio data to find, respectively, the drum pattern of the song and whale harmonization \cite{keoghMP}. Finally, \textsc{El\_load}, \textsc{Quake} and \textsc{Ltmm} were chosen based on the domain information.

Furthermore, we characterize each dataset by the \emph{contrast} $c_{d,1|n}$ of its top $d$-dimensional motif, as per Definition~\ref{sec:contrast}.
Datasets with a small contrast are expected to be more difficult, i.e. require more time to find the top motif.
For instance \textsc{El\_load} is expected to be more difficult compared to \textsc{Weather}, since the latter has a pair of subsequences that repeat the same shape at a distance that is significantly smaller than the $n$-th closest pair.

\subsubsection*{Default parameter values}
\textsc{MOMENTI} will have its parameters set to default unless otherwise indicated.
The failure probability is set at $\delta=0.01$, the maximum hash length is set to $K = 8$, while the maximum number of repetitions is set to $L = 200$. $r$ is automatically estimated using the heuristic introduced in \Cref{sec:r_auto}.

\subsection{Finding the top motif}
\begin{table}[t]
\caption{\added{Time required to find the $d$-dimensional top motif, for fixed $d$, averaged over 9 runs. Values in parentheses are estimates.}
}
\centering
\begin{tabular}{lrrrrrr}
\toprule
dataset    &\multicolumn{2}{c}{MOMENTI} & MSTUMP    & EMD & RP    \\ 
\cmidrule{2-3}
 & Index build & Total  &&\\
\midrule
potentials & 0.11 & \cellcolor[HTML]{CDF2D9} 0.51 &   3.65     &   4.80 & 3.20 \\
evaporator & 0.16 & \cellcolor[HTML]{CDF2D9} 0.55 &  4.45 &    12.95 & 6.78\\
RUTH &  2.91 & \cellcolor[HTML]{CDF2D9} 8.10 &   84.04 & 1.5h & 2.3h\\
weather   & 15.04 & \cellcolor[HTML]{CDF2D9} 33.37 &   1035.73    &          - & 1.2h \\ 
whales   & 60.67 & \cellcolor[HTML]{CDF2D9} 2.2 h   &    (2.7 days) & - & - \\ 
quake    & 175.3 & \cellcolor[HTML]{CDF2D9} 3.6 h &  (7.2 days) & - & - \\ 
el\_load    & 180.2 & \cellcolor[HTML]{CDF2D9} 2.8 h &   (8.4 days) & - & - \\ 
LTMM    & 240.6 & \cellcolor[HTML]{CDF2D9} 15.6 h &   (11.8 days)    &  - & - \\
\bottomrule
\end{tabular}
\label{tab:time-fixed-d}
\end{table}

In this first experiment, the task is to find the top motif at a given dimensionality $d$ (as per Table~\ref{tab:dataset}).
\replaced{For each dataset we report the total running time of the baselines \textsc{Mstump}, \textsc{EMD} and \textsc{RP} in the last three columns, whereas for our approach we report both the total time and the time required to set up the index.}{}
All times are in seconds, unless otherwise noted. Furthermore, for timed out runs of \textsc{Mstump} we report, in parentheses, an estimate of the running time, which is made possible by the very regular behavior of \textsc{Mstump} with respect to the input size. For \textsc{EDM} \added{and \textsc{RP}} the running time is more unpredictable, therefore we refrain from providing estimates for timed out runs.

As can be seen, our algorithm is faster on all settings, completing the execution orders of magnitude faster than the baselines on the larger datasets.
To substantiate this observation, Table~\ref{tab:num_dist} reports the number of distance computations carried out by each approach:
\textsc{MOMENTI} usage of LSH allows to effectively prune most distance computations, whereas \textsc{Mstump} computes a quadratic number of distances.

\begin{table}[t]
\caption{{Number of cumulative distance computations to find the top $d$-dimensional motif.
}}
\begin{tabular}{lrrrr}
\hline
dataset    & MOMENTI & MSTUMP & EMD & RP \\ \hline
potentials & $64$&  \addstackgap[1.5pt]{ $2.4 \cdot 10^7$ }& $1.0\cdot 10^2$ & $2.7 \cdot 10^3$\\
evaporator & $4.7\cdot 10^2$&   $1.4\cdot 10^8$ & 5 & $4.4\cdot 10^4$\\
RUTH       & $1.2 \cdot 10^2$  &  $3.3 \cdot 10^9$ & $5.4\cdot 10^3$ &  $2.3\cdot 10^5$\\
weather    & $9.8 \cdot 10^4$   &   $3.9 \cdot 10^{10}$ & - & $2.1 \cdot 10^ 6$\\
whales    & $3.1 \cdot 10^8$   &   $(3.2 \cdot 10^{12})$ & - & - \\
quake    & $1.1 \cdot 10^6$   &   $(6.6\cdot10^{14})$ & - & - \\
el\_load & $4.4 \cdot 10^2$ & $(2.4\cdot 10^{14})$ & - & - \\
LTMM & $4.4 \cdot 10^2$ & $(1.9\cdot 10^{15})$ & - & -\\\hline
\end{tabular}
\label{tab:num_dist}
\end{table}

The \textsc{EMD} algorithm appears to be the slowest, its major drawback being the set up of the parameters for a successful discovery. \added{On the other hand, \textsc{RP} scales better to larger datasets but its parameter tuning heuristic repeatedly discards the collision matrix when a sparsification target is not met, this is  especially costly on high dimensional datasets.
} 
We stress that the Matrix Profile computed by \textsc{Mstump} allows the discovery of motifs of any dimensionality, while in this test \textsc{MOMENTI} only finds the top motif for a fixed dimensionality. In the next section
we will investigate how \textsc{MOMENTI} behaves when finding motifs for all dimensionalities at the same time.

Finally, Table~\ref{tab:space} reports the memory usage, in gigabytes, of the different algorithms.
We observe that \textsc{MOMENTI} requires the least memory to execute.
\paragraph*{Quality of the motifs}
\added{
To better understand what pairs are returned by the different algorithms
we measured the Mean Absolute Relative Error (MARE) of the distances of the discovered motifs with respect to the exact ones. We found that \textsc{RP} can reach MARE values as high as 40\%, while \textsc{EMD} stays below 20\%. In contrast, our method consistently maintains MARE under 3\%, underlining how beneficial LSH, that mainly explores close pairs, is in this context. Further details on the quality of our results will be given in \Cref{sec:quality}.}

\begin{table}[t]
\caption{\textbf{Memory required to find the top $d$-dimensional motif.}}
\centering
\begin{tabular}{lrrrrr}
\hline
& \multicolumn{4}{c}{Space (GB)} & \\ \cline{2-5} 
dataset    & MOMENTI       & MSTUMP    & EMD & RP \\ \hline
potentials &    \cellcolor[HTML]{CDF2D9}     0.016&        0.028&           0.79 & 0.07\\
evaporator &    \cellcolor[HTML]{CDF2D9}     0.020&        0.030&           0.81 & 0.1\\
RUTH & \cellcolor[HTML]{CDF2D9} 0.025&        0.084&           1.19 & 0.5 \\
weather    &  \cellcolor[HTML]{CDF2D9}  0.027&       0.106&           - & 0.8\\ 
whales    &  \cellcolor[HTML]{CDF2D9}  0.24&       3.2&           - & -\\ 
quake    &  \cellcolor[HTML]{CDF2D9}  2.13&       23.0&           - & -\\ 
el\_load    &  \cellcolor[HTML]{CDF2D9} 1.50      &    7.4     &     -      & -\\ 
LTMM    &  \cellcolor[HTML]{CDF2D9}  3.00&       15.7&           - & -\\
\hline
\end{tabular}
\label{tab:space}
\end{table}
\subsection{Finding the top motifs with different dimensionalities}

\begin{table}[t]
\caption{Time required for the top motifs for all dimensionalities. Mean over 9 runs. Values in parentheses are estimates.
}
\centering
\begin{tabular}{lrr}
\toprule
dataset    & MOMENTI  & MSTUMP      \\
\midrule
potentials &  0.75  & 3.65   \\
evaporator &  0.89  & 4.45 \\
RUTH &  12.46  & 84.04 \\
weather   &  47.85   & 1035.73 \\ 
whales   &  2.5 h  &  (2.7 days) \\ 
quake    & 5.8 h  &(7.2 days)\\ 
el\_load    &  3.1 h   & (8.4 days)\\ 
LTMM    & 16 h & (11.8 days) \\
\bottomrule
\end{tabular}
\label{tab:time-multi-d}
\end{table}

\begin{figure}
    \centering
    \includegraphics[width=\linewidth]{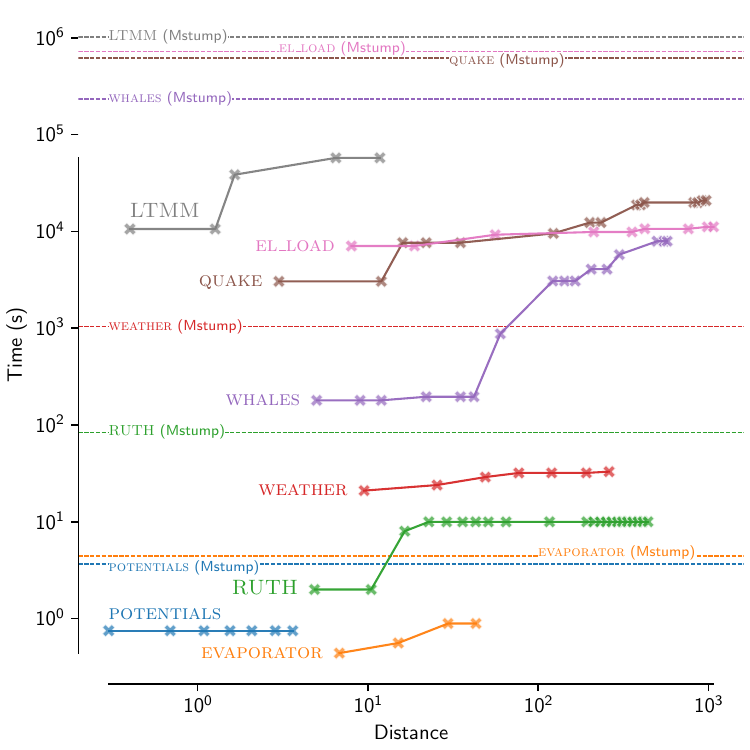}
    \caption{Solid lines mark the time required by \textsc{MOMENTI} to find the top motif of each dimensionality for all datasets; dashed lines mark the time required by the \textsc{Mstump} baseline for the same task.}
    \label{fig:multisub}
    \Description[Time required for multidimensional motif discovery in each dataset. The times are plotted with respect to the distance. Greater distances require more time.]{}
\end{figure}

In this second experiment the aim is to find the motifs for all dimensionalities $2\le d \le D$ in the same execution.
MSTUMP does this natively, whereas our algorithm \textsc{MOMENTI} can be adapted to do so as described in Section~\ref{sec:multisub}. \textsc{EMD} \added{and \textsc{RP}} can only achieve this through multiple executions and \replaced{are}{} thus excluded from this experiment.
In Table~\ref{tab:time-multi-d} we report the results for this experiment.
Note that the running times for \textsc{Mstump} are identical to the ones reported in Table~\ref{tab:time-fixed-d}, given that \textsc{Mstump} finds motifs for all dimensionalities natively.
As for \textsc{MOMENTI} the indexing time is the same in this setting as the one reported in Table~\ref{tab:time-fixed-d}, therefore we omit it for clarity.
We observe that even in this more challenging scenario our algorithm is able to discover the motifs for all dimensionalities faster than the \textsc{Mstump} baseline.
Indeed, the running time of \textsc{MOMENTI} is always within a factor $\approx 1.6$ of its running time on the one dimensional case.

To further investigate this behavior, in Figure~\ref{fig:multisub} we report the relation between the discovery time of a motif ($x$ axis) and its distance ($y$ axis). Each color identifies a dataset, and each point represents a motif for some dimensionality: for instance the \textsc{Weather} dataset (orange) has 7 points in this figure because it has 8 dimensions, and we set the algorithms to find the 7 motifs spanning between 2 and 8 dimensions.
The dashed horizontal lines mark the runtime performance of \textsc{Mstump}, which reports all the motifs at the same time at the end of its execution.
For instance, the orange dashed horizontal line reports that \textsc{Mstump} takes $\approx 10^3$ seconds on the \textsc{Weather} dataset.
The times marked by dashed lines and by the rightmost points of each solid line are the same as in Table~\ref{tab:time-multi-d}.

As can be observed in \Cref{fig:multisub}, by virtue of how the algorithm runs, motifs at a shorter distance are found earlier.
For instance, for \textsc{Weather} the motif of dimensionality $4$ has a $\operatorname{dist}_4\simeq49$ while the motif of dimensionality $5$ has $\operatorname{dist}_5\simeq 77$, which requires more iterations (and thus more time) to meet the stopping condition of the algorithm.
As we discussed earlier, this fact can be used in an \emph{anytime} fashion:
the execution can be stopped at any point after the first motifs have been returned, knowing that motifs yet to be found are at larger distances and thus might be uninteresting.
\subsection{Scalability}
\begin{figure}
    \centering
    \includegraphics[width=\linewidth]{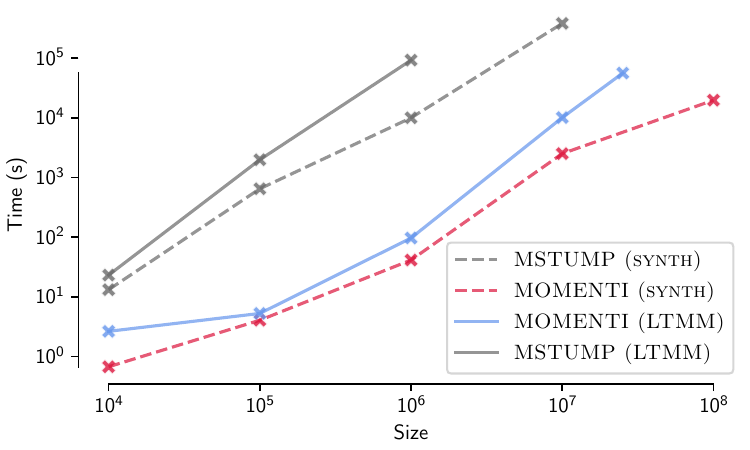}
    \caption{ Scalability vs. input size (log scale).}
    \label{fig:scalability}
    \Description[Subquadratic behaviour of \textsc{MOMENTI} with respect to \textsc{Mstump} on different input sizes.]{}
\end{figure}

We now test the scalability of \textsc{MOMENTI} compared to \textsc{Mstump}. \added{We omit EMD and RP from this comparison, as previous sections show that they do not scale beyond moderately sized time series.}
While the behavior of \textsc{Mstump} is data-independent, i.e. its running time depends only on the size of the data,
for \textsc{MOMENTI} the running time is influenced by both the size of the time series and the contrast of the motifs.
\replaced{Therefore, to test the scalability in a robust way we employ our largest dataset and a synthetic one. We generate a random walk with $D=5$ of length between 10000 and 100 million, planting a sinusoidal motif of length $w=300$ across $d=2$ dimensions perturbed with Gaussian noise such that $c_{2,1|n}$ is 1.1. For \textsc{LTMM} we pick chunks of increasing size centered around the true motif, we underline that each chunk has a different contrast, so some could be harder than others for our algorithm.
}{}
The results are reported in \Cref{fig:scalability}, where both axes use a logarithmic scale.
We observe that our algorithm scales better with the input size, with a sub-quadratic running time.

\subsection{Influence of parameters on the running time}
We now study the impact of the parameters $K$, $L$, and $r$ on the running time, that in the previous experiments were either fixed ($K=8$ and $L=200$) or estimated from the data (the quantization parameter $r$).
Remember that our algorithm finds the correct answer with probability $1-\delta$ for any setting of the parameters, that influence only the performance.
\replaced{Each experiment in the following has been repeated 10 times: the plots report the average as a colored line, additionally, a gray dashed line reports the sizes of the hash indices built under different settings. Dotted lines indicate that the result is estimated due to timeout.}{}
\paragraph{Impact of concatenations $K$}
\begin{figure}[t]
    \centering
    \includegraphics[width=\linewidth]{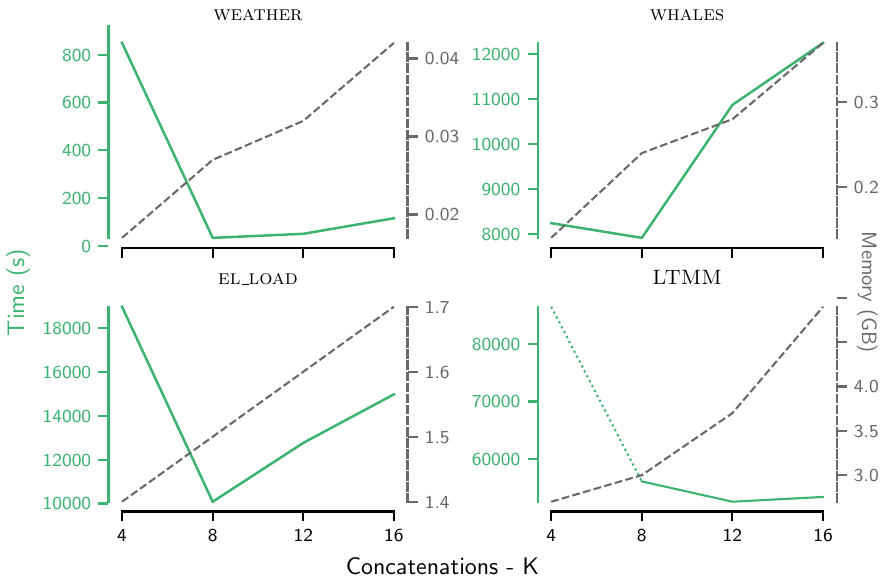}
    \caption{ \added{Time and space requirements for motif discovery at different maximum allowed values of $K$.}}
    \label{fig:K_graphs}
    \Description[All datasets show a minimum around the value 8.]{}
\end{figure}
We test $K\in\{4,8,12,16\}$, reporting the results in \Cref{fig:K_graphs}.
We observe that a higher number of concatenations is related to longer computing times,
both because more repetition of the outer loop of Algorithm~\ref{alg:emitaggr} have to be executed and because the index construction takes longer.
Conversely, using short hash values with $K=4$ incurs high execution times as well, mainly because fewer distance computations are pruned.
In all tested cases the best tradeoff is achieved with $K=8$, which is the recommended value and the one we used in all previous experiments.
We stress that for any value of $K$ considered in this section \textsc{MOMENTI} is faster than the baseline \textsc{Mstump}.

\begin{figure}[t]
    \centering
    \includegraphics[width=\linewidth]{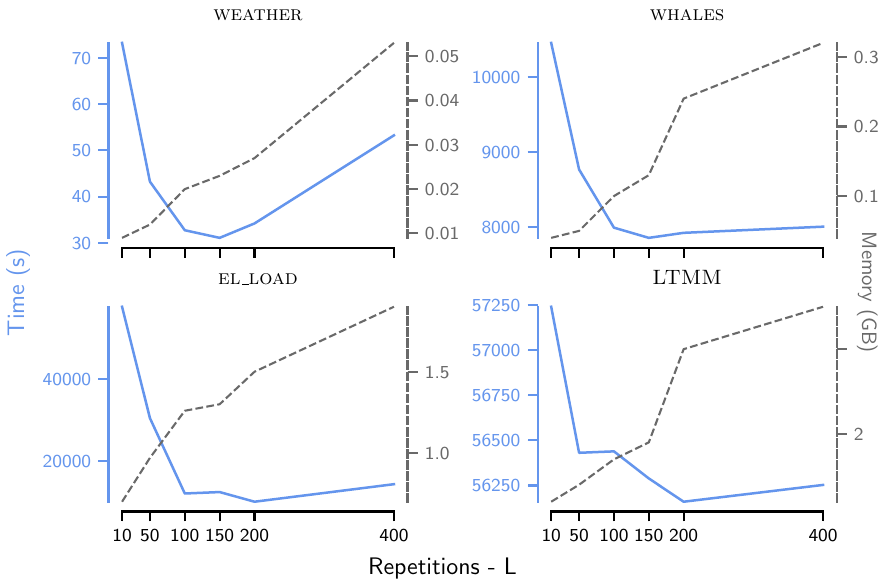}
    \caption{ \added{Time and space requirements for motif discovery at different maximum allowed values of $L$.}}
    \label{fig:l_graphs}
    \Description[Both a small and a large value of L requires more search time than the middle values]{}
\end{figure}
    
\begin{figure}[t]
    \centering
    \includegraphics[width=\linewidth]{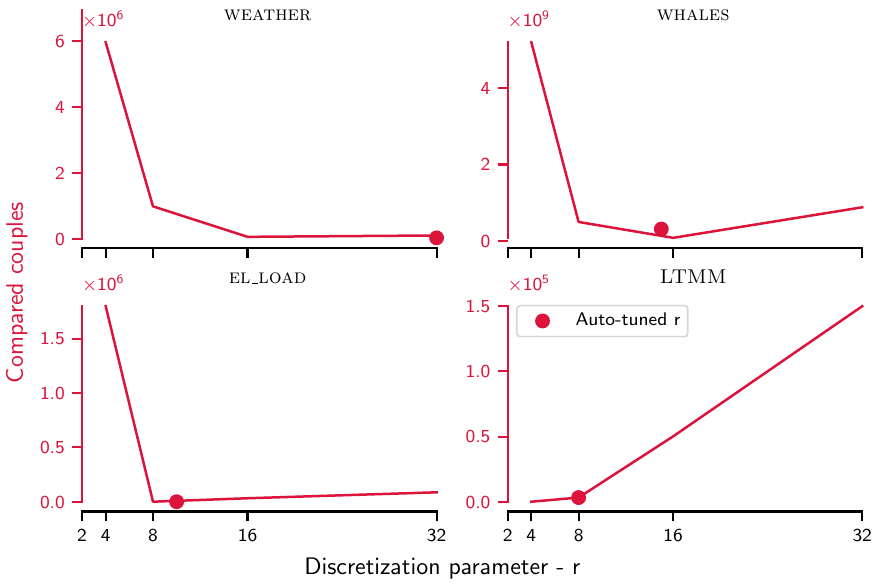}
    \caption{ \textbf{Number of comparisons for different values of $r$.} The red dot is the value found by the heuristic in \Cref{sec:r_auto}.}
    \label{fig:r_graphs}
    \Description[Each dataset shows a different behavior with respect to the r values, the one found by our heuristic is always close to the empirical minima.]{}
\end{figure}

\begin{figure}
    \centering
    \includegraphics[width=\linewidth]{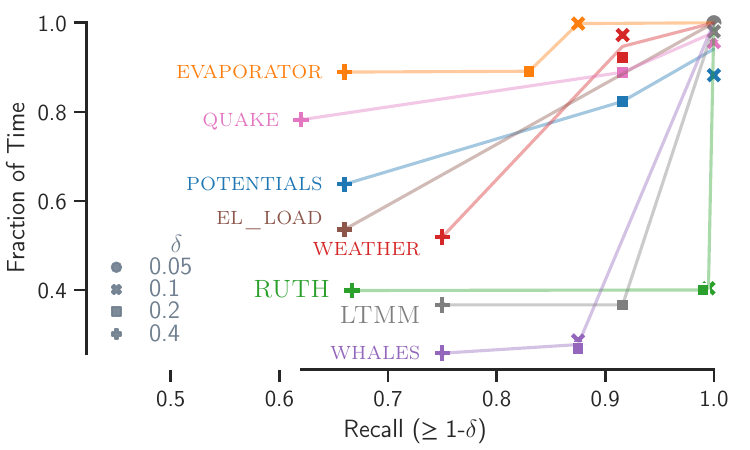}
    \caption{\added{Relation between $\delta$, represented by the markers, the measured recall and the fraction of time required with respect to the tighter bound.} }
    \label{fig:time_delta}
    \Description[All datasets have an higher measured recall with respect to the requested one. At the same time datasets RUTH, LTMM and Whales that share a high contrast behave in a similar way, at lower requested deltas the fraction of time required for the algorithm to terminate is significantly lower.]{}
\end{figure}
\paragraph{Impact of repetitions L}
We test $L$ in the range from 10 to 400, setting $K=8$:
a large $L$ will allow the stopping condition to be satisfied in earlier iteration of the outer loop of Algorithm~\ref{alg:emitaggr}, thus requiring fewer distance computations, at the expense of a higher index build time.
\Cref{fig:l_graphs} explores this trade-off, with values of $L$ between $100$ and $200$ minimizing the search time. \added{Adding more repetitions increases the memory usage, which however remains rather moderate, with $3.5$Gb being used for the largest dataset, LTMM. }
Datasets that benefit more from an increased number of repetitions are those whose relative contrast is higher (e.g. \textsc{LTMM}), because motifs are more likely to share long hash prefix.
\added{It is important to note that the sublinear behavior in memory growth is due to our use of the tensoring implementation.}
\paragraph{Impact of the quantization parameter $r$}
In this experiment, we manually fix $r\in\{4,8,12,16\}$ and compare it with the value automatically chosen by the algorithm using the procedure described in \Cref{sec:r_auto} in terms of the resulting number of distance computations, and hence running time.
Figure~\ref{fig:r_graphs} shows that the choice of $r$ has a dramatic impact on the number of distance computations. 
Remarkably, the value automatically picked by our algorithm in a data-dependent way (and used in the rest of the experiments we presented) attains the same performance as the best fixed parameter considered in this experiment.

\subsection{Quality of the results}
\label{sec:quality}
\paragraph{Impact of probability threshold $\delta$}
\added{
We test the influence of $\delta$ (that controls the algorithm's failure rate) in the range from $0.05$ to $0.4$: larger values allow the algorithm to satisfy the stopping condition earlier, improving runtime at the cost of reduced recall.
For each setting we measured the achieved recall, the relative runtime with respect to the stricter setting ($\delta=0.05$), and the mean absolute relative error (MARE) of distances of the returned motifs with respect to the ground truth.
In \Cref{fig:time_delta} we present the recall-efficiency trade-off.
We observe that the target recall imposed with $\delta$ is always satisfied, moreover, the fraction of time required with respect to our highest target is stable under similar contrast values (note the similar behavior of \textsc{LTMM}, \textsc{Weather} and \textsc{Ruth}, which share a relatively high contrast).
About the quality of the discovered motifs at lower deltas, we found that the MARE never exceeds $3\%$ for $\delta \leq 0.2$, indicating that the quality of the result is still preserved.
}

\paragraph{Impact of noise}
\added{In this last experiment we replicate the experimental design of~\cite{keoghMP},
where the task is to reliably find motifs even when datasets are cluttered by noisy dimensions.}
\added{To evaluate the robustness of our algorithm, we added from 4 to 256 additional dimensions, each being a random walk, to datasets \textsc{Potentials}, \textsc{Evaporator}, \textsc{Ruth}, \textsc{Weather}. Then we ran \textsc{MOMENTI} looking for the $d$-dimensional motif, with $d$ as per Table~\ref{tab:dataset}, repeating each experiment 12 times. In all settings, irrespective of the number of additional noisy dimensions, our algorithm attained recall values of 1, showing its robustness to the number of irrelevant dimensions.}

\section{Conclusions}
We presented a LSH based algorithm for the motif discovery problem in multidimensional time series, with strong guarantees on the quality of the results.
Experimental evidence supports the efficacy and efficiency of our approach, demonstrating that it is able to find motifs with only a fraction of all possible distance computations.

\balance
\begin{acks}
This work was supported in part by the Big-Mobility project by
the University of Padova under the Uni-Impresa call, by the MUR
PRIN 2022TS4Y3N EXPAND project, by MUR PNRR CN00000013
National Center for HPC, Big Data and Quantum Computing, and by  Marsden Fund (MFP-UOA2226).
\end{acks}

\clearpage

\bibliographystyle{ACM-Reference-Format}
\bibliography{lib}


\begin{thebibliography}{44}


\ifx \showCODEN    \undefined \def \showCODEN     #1{\unskip}     \fi
\ifx \showDOI      \undefined \def \showDOI       #1{#1}\fi
\ifx \showISBNx    \undefined \def \showISBNx     #1{\unskip}     \fi
\ifx \showISBNxiii \undefined \def \showISBNxiii  #1{\unskip}     \fi
\ifx \showISSN     \undefined \def \showISSN      #1{\unskip}     \fi
\ifx \showLCCN     \undefined \def \showLCCN      #1{\unskip}     \fi
\ifx \shownote     \undefined \def \shownote      #1{#1}          \fi
\ifx \showarticletitle \undefined \def \showarticletitle #1{#1}   \fi
\ifx \showURL      \undefined \def \showURL       {\relax}        \fi
\providecommand\bibfield[2]{#2}
\providecommand\bibinfo[2]{#2}
\providecommand\natexlab[1]{#1}
\providecommand\showeprint[2][]{arXiv:#2}

\bibitem[\protect\citeauthoryear{Abanda, Mori, and Lozano}{Abanda et~al\mbox{.}}{2019}]%
        {distances}
\bibfield{author}{\bibinfo{person}{Amaia Abanda}, \bibinfo{person}{Usue Mori}, {and} \bibinfo{person}{Jose~A. Lozano}.} \bibinfo{year}{2019}\natexlab{}.
\newblock \showarticletitle{A review on distance based time series classification}.
\newblock \bibinfo{journal}{\emph{Data Min. Knowl. Discov.}} \bibinfo{volume}{33}, \bibinfo{number}{2} (\bibinfo{date}{March} \bibinfo{year}{2019}), \bibinfo{pages}{378–412}.
\newblock
\showISSN{1384-5810}
\urldef\tempurl%
\url{https://doi.org/10.1007/s10618-018-0596-4}
\showDOI{\tempurl}


\bibitem[\protect\citeauthoryear{Andoni and Indyk}{Andoni and Indyk}{2006}]%
        {rhoboundlsh}
\bibfield{author}{\bibinfo{person}{Alexandr Andoni} {and} \bibinfo{person}{Piotr Indyk}.} \bibinfo{year}{2006}\natexlab{}.
\newblock \showarticletitle{Near-Optimal Hashing Algorithms for Approximate Nearest Neighbor in High Dimensions}. In \bibinfo{booktitle}{\emph{FOCS}}. \bibinfo{pages}{459--468}.
\newblock
\urldef\tempurl%
\url{https://doi.org/10.1109/FOCS.2006.49}
\showDOI{\tempurl}


\bibitem[\protect\citeauthoryear{Andoni and Indyk}{Andoni and Indyk}{2008}]%
        {lsh}
\bibfield{author}{\bibinfo{person}{Alexandr Andoni} {and} \bibinfo{person}{Piotr Indyk}.} \bibinfo{year}{2008}\natexlab{}.
\newblock \showarticletitle{Near-optimal hashing algorithms for approximate nearest neighbor in high dimensions}.
\newblock \bibinfo{journal}{\emph{Commun. ACM}} \bibinfo{volume}{51}, \bibinfo{number}{1} (\bibinfo{year}{2008}), \bibinfo{pages}{117–122}.
\newblock
\showISSN{0001-0782}
\urldef\tempurl%
\url{https://doi.org/10.1145/1327452.1327494}
\showDOI{\tempurl}


\bibitem[\protect\citeauthoryear{Arnu, Yaqub, Mocci, Colla, Neuer, Fricout, Renard, Mozzati, and Gallinari}{Arnu et~al\mbox{.}}{2017}]%
        {SteelMotifs}
\bibfield{author}{\bibinfo{person}{David Arnu}, \bibinfo{person}{Edwin Yaqub}, \bibinfo{person}{Claudio Mocci}, \bibinfo{person}{Valentina Colla}, \bibinfo{person}{Marcus Neuer}, \bibinfo{person}{Gabriel Fricout}, \bibinfo{person}{Xavier Renard}, \bibinfo{person}{Christophe Mozzati}, {and} \bibinfo{person}{Patrick Gallinari}.} \bibinfo{year}{2017}\natexlab{}.
\newblock \showarticletitle{A reference architecture for quality improvement in steel production}. In \bibinfo{booktitle}{\emph{Data Science--Analytics and Applications: Proceedings of the 1st International Data Science Conference--iDSC2017}}. Springer, \bibinfo{pages}{85--90}.
\newblock


\bibitem[\protect\citeauthoryear{Aum\"{u}ller, Christiani, Pagh, and Vesterli}{Aum\"{u}ller et~al\mbox{.}}{2019}]%
        {puffinn}
\bibfield{author}{\bibinfo{person}{Martin Aum\"{u}ller}, \bibinfo{person}{Tobias Christiani}, \bibinfo{person}{Rasmus Pagh}, {and} \bibinfo{person}{Michael Vesterli}.} \bibinfo{year}{2019}\natexlab{}.
\newblock \showarticletitle{{PUFFINN: Parameterless and Universally Fast FInding of Nearest Neighbors}}. In \bibinfo{booktitle}{\emph{ESA}} \emph{(\bibinfo{series}{LIPIcs})}, \bibfield{editor}{\bibinfo{person}{Michael~A. Bender}, \bibinfo{person}{Ola Svensson}, {and} \bibinfo{person}{Grzegorz Herman}} (Eds.), Vol.~\bibinfo{volume}{144}. \bibinfo{publisher}{Schloss Dagstuhl -- Leibniz-Zentrum f{\"u}r Informatik}, \bibinfo{address}{Dagstuhl, Germany}, \bibinfo{pages}{10:1--10:16}.
\newblock
\showISBNx{978-3-95977-124-5}
\showISSN{1868-8969}
\urldef\tempurl%
\url{https://doi.org/10.4230/LIPIcs.ESA.2019.10}
\showDOI{\tempurl}


\bibitem[\protect\citeauthoryear{Balasubramanian, Wang, and Prabhakaran}{Balasubramanian et~al\mbox{.}}{2016}]%
        {Healthcare}
\bibfield{author}{\bibinfo{person}{Arvind Balasubramanian}, \bibinfo{person}{Jun Wang}, {and} \bibinfo{person}{Balakrishnan Prabhakaran}.} \bibinfo{year}{2016}\natexlab{}.
\newblock \showarticletitle{Discovering Multidimensional Motifs in Physiological Signals for Personalized Healthcare}.
\newblock \bibinfo{journal}{\emph{IEEE Journal of Selected Topics in Signal Processing}} \bibinfo{volume}{10}, \bibinfo{number}{5} (\bibinfo{year}{2016}), \bibinfo{pages}{832--841}.
\newblock
\urldef\tempurl%
\url{https://doi.org/10.1109/JSTSP.2016.2543679}
\showDOI{\tempurl}


\bibitem[\protect\citeauthoryear{Berlin and Van~Laerhoven}{Berlin and Van~Laerhoven}{2012a}]%
        {ActivityDiscovery2}
\bibfield{author}{\bibinfo{person}{Eugen Berlin} {and} \bibinfo{person}{Kristof Van~Laerhoven}.} \bibinfo{year}{2012}\natexlab{a}.
\newblock \showarticletitle{Detecting leisure activities with dense motif discovery}. In \bibinfo{booktitle}{\emph{Proceedings of the 2012 ACM Conference on Ubiquitous Computing}} (Pittsburgh, Pennsylvania) \emph{(\bibinfo{series}{UbiComp '12})}. \bibinfo{publisher}{Association for Computing Machinery}, \bibinfo{address}{New York, NY, USA}, \bibinfo{pages}{250–259}.
\newblock
\showISBNx{9781450312240}
\urldef\tempurl%
\url{https://doi.org/10.1145/2370216.2370257}
\showDOI{\tempurl}


\bibitem[\protect\citeauthoryear{Berlin and Van~Laerhoven}{Berlin and Van~Laerhoven}{2012b}]%
        {berlin2012detecting}
\bibfield{author}{\bibinfo{person}{Eugen Berlin} {and} \bibinfo{person}{Kristof Van~Laerhoven}.} \bibinfo{year}{2012}\natexlab{b}.
\newblock \showarticletitle{Detecting leisure activities with dense motif discovery}. In \bibinfo{booktitle}{\emph{ACM Conference on Ubiquitous Computing}}. \bibinfo{publisher}{ACM}, \bibinfo{address}{New York, NY, USA}, \bibinfo{pages}{250–259}.
\newblock
\showISBNx{9781450312240}
\urldef\tempurl%
\url{https://doi.org/10.1145/2370216.2370257}
\showDOI{\tempurl}


\bibitem[\protect\citeauthoryear{Berthold and Höppner}{Berthold and Höppner}{2016}]%
        {berthold2016pearseucl}
\bibfield{author}{\bibinfo{person}{Michael~R. Berthold} {and} \bibinfo{person}{Frank Höppner}.} \bibinfo{year}{2016}\natexlab{}.
\newblock \bibinfo{title}{On Clustering Time Series Using Euclidean Distance and Pearson Correlation}.
\newblock
\newblock
\showeprint[arxiv]{1601.02213}~[cs.LG]
\urldef\tempurl%
\url{https://arxiv.org/abs/1601.02213}
\showURL{%
\tempurl}


\bibitem[\protect\citeauthoryear{Cassisi, Aliotta, Cannata, Montalto, Patanè, Pulvirenti, and Spampinato}{Cassisi et~al\mbox{.}}{2012}]%
        {2012_Cassisi}
\bibfield{author}{\bibinfo{person}{Carmelo Cassisi}, \bibinfo{person}{Marco Aliotta}, \bibinfo{person}{Andrea Cannata}, \bibinfo{person}{Placido Montalto}, \bibinfo{person}{Domenico Patanè}, \bibinfo{person}{Alfredo Pulvirenti}, {and} \bibinfo{person}{Letizia Spampinato}.} \bibinfo{year}{2012}\natexlab{}.
\newblock \showarticletitle{Motif Discovery on Seismic Amplitude Time Series: The Case Study of Mt Etna 2011 Eruptive Activity}.
\newblock \bibinfo{journal}{\emph{Pure and Applied Geophysics}} \bibinfo{volume}{170}, \bibinfo{number}{4} (\bibinfo{year}{2012}), \bibinfo{pages}{529--545}.
\newblock
\urldef\tempurl%
\url{https://doi.org/10.1007/s00024-012-0560-y}
\showDOI{\tempurl}


\bibitem[\protect\citeauthoryear{Ceccarello and Gamper}{Ceccarello and Gamper}{2022}]%
        {attimo}
\bibfield{author}{\bibinfo{person}{Matteo Ceccarello} {and} \bibinfo{person}{Johann Gamper}.} \bibinfo{year}{2022}\natexlab{}.
\newblock \showarticletitle{Fast and Scalable Mining of Time Series Motifs with Probabilistic Guarantees}.
\newblock \bibinfo{journal}{\emph{Proc. VLDB Endow.}} \bibinfo{volume}{15}, \bibinfo{number}{13} (\bibinfo{date}{Sept.} \bibinfo{year}{2022}), \bibinfo{pages}{3841–3853}.
\newblock
\showISSN{2150-8097}
\urldef\tempurl%
\url{https://doi.org/10.14778/3565838.3565840}
\showDOI{\tempurl}


\bibitem[\protect\citeauthoryear{Christiani}{Christiani}{2019}]%
        {christianitensoring}
\bibfield{author}{\bibinfo{person}{Tobias Christiani}.} \bibinfo{year}{2019}\natexlab{}.
\newblock \showarticletitle{Fast Locality-Sensitive Hashing Frameworks for Approximate Near Neighbor Search}. In \bibinfo{booktitle}{\emph{SISAP Proceedings}}. \bibinfo{publisher}{Springer-Verlag}, \bibinfo{address}{Berlin, Heidelberg}, \bibinfo{pages}{3–17}.
\newblock
\showISBNx{978-3-030-32046-1}
\urldef\tempurl%
\url{https://doi.org/10.1007/978-3-030-32047-8_1}
\showDOI{\tempurl}


\bibitem[\protect\citeauthoryear{Christiani, Pagh, and Thorup}{Christiani et~al\mbox{.}}{2020}]%
        {confsampling}
\bibfield{author}{\bibinfo{person}{Tobias Christiani}, \bibinfo{person}{Rasmus Pagh}, {and} \bibinfo{person}{Mikkel Thorup}.} \bibinfo{year}{2020}\natexlab{}.
\newblock \showarticletitle{Confirmation Sampling for Exact Nearest Neighbor Search}. In \bibinfo{booktitle}{\emph{SISAP Proceedings}}. \bibinfo{publisher}{Springer-Verlag}, \bibinfo{address}{Berlin, Heidelberg}, \bibinfo{pages}{97–110}.
\newblock
\showISBNx{978-3-030-60935-1}
\urldef\tempurl%
\url{https://doi.org/10.1007/978-3-030-60936-8_8}
\showDOI{\tempurl}


\bibitem[\protect\citeauthoryear{Datar, Immorlica, Indyk, and Mirrokni}{Datar et~al\mbox{.}}{2004}]%
        {datar2004}
\bibfield{author}{\bibinfo{person}{Mayur Datar}, \bibinfo{person}{Nicole Immorlica}, \bibinfo{person}{Piotr Indyk}, {and} \bibinfo{person}{Vahab~S. Mirrokni}.} \bibinfo{year}{2004}\natexlab{}.
\newblock \showarticletitle{Locality-sensitive hashing scheme based on p-stable distributions}. In \bibinfo{booktitle}{\emph{Proceedings of SCG}}. \bibinfo{publisher}{ACM}, \bibinfo{address}{New York, NY, USA}, \bibinfo{pages}{253–262}.
\newblock
\showISBNx{1581138857}
\urldef\tempurl%
\url{https://doi.org/10.1145/997817.997857}
\showDOI{\tempurl}


\bibitem[\protect\citeauthoryear{Daw, Finney, and Tracy}{Daw et~al\mbox{.}}{2003}]%
        {symbolizationreview}
\bibfield{author}{\bibinfo{person}{C.~S. Daw}, \bibinfo{person}{C.~E.~A. Finney}, {and} \bibinfo{person}{E.~R. Tracy}.} \bibinfo{year}{2003}\natexlab{}.
\newblock \showarticletitle{{A review of symbolic analysis of experimental data}}.
\newblock \bibinfo{journal}{\emph{Review of Scientific Instruments}} \bibinfo{volume}{74}, \bibinfo{number}{2} (\bibinfo{date}{02} \bibinfo{year}{2003}), \bibinfo{pages}{915--930}.
\newblock
\showISSN{0034-6748}
\urldef\tempurl%
\url{https://doi.org/10.1063/1.1531823}
\showDOI{\tempurl}


\bibitem[\protect\citeauthoryear{De~Moor}{De~Moor}{2024a}]%
        {DaISyPREG}
\bibfield{author}{\bibinfo{person}{B.L.R.~(ed.) De~Moor}.} \bibinfo{year}{2024}\natexlab{a}.
\newblock \bibinfo{booktitle}{\emph{{DaISy}: Database for the Identification of Systems, Cutaneous potential recordings of a pregnant woman, Biomedical Systems, 96-012}}.
\newblock Department of Electrical Engineering, ESAT/STADIUS, KU Leuven, Belgium.
\newblock
\urldef\tempurl%
\url{http://homes.esat.kuleuven.be/~smc/daisy/}
\showURL{%
Retrieved September 3, 2024 from \tempurl}


\bibitem[\protect\citeauthoryear{De~Moor}{De~Moor}{2024b}]%
        {DaISyEVAP}
\bibfield{author}{\bibinfo{person}{B.L.R.~(ed.) De~Moor}.} \bibinfo{year}{2024}\natexlab{b}.
\newblock \bibinfo{booktitle}{\emph{{DaISy}: Database for the Identification of Systems, Data from an industrial evaporator, Process Industry Systems, 96-010.}}
\newblock Department of Electrical Engineering, ESAT/STADIUS, KU Leuven, Belgium.
\newblock
\urldef\tempurl%
\url{http://homes.esat.kuleuven.be/~smc/daisy/}
\showURL{%
Retrieved September 3, 2024 from \tempurl}


\bibitem[\protect\citeauthoryear{De~Paepe, Avendano, and Van~Hoecke}{De~Paepe et~al\mbox{.}}{2020}]%
        {de2019implications}
\bibfield{author}{\bibinfo{person}{Dieter De~Paepe}, \bibinfo{person}{Diego~Nieves Avendano}, {and} \bibinfo{person}{Sofie Van~Hoecke}.} \bibinfo{year}{2020}\natexlab{}.
\newblock \showarticletitle{Implications of Z-Normalization in the Matrix Profile}. In \bibinfo{booktitle}{\emph{Pattern Recognition Applications and Methods}}, \bibfield{editor}{\bibinfo{person}{Maria De~Marsico}, \bibinfo{person}{Gabriella Sanniti~di Baja}, {and} \bibinfo{person}{Ana Fred}} (Eds.). \bibinfo{publisher}{Springer International Publishing}, \bibinfo{address}{Cham}, \bibinfo{pages}{95--118}.
\newblock
\showISBNx{978-3-030-40014-9}


\bibitem[\protect\citeauthoryear{Godahewa, Bergmeir, Webb, Hyndman, and Montero-Manso}{Godahewa et~al\mbox{.}}{2021}]%
        {godahewa_2021_5184708}
\bibfield{author}{\bibinfo{person}{Rakshitha Godahewa}, \bibinfo{person}{Christoph Bergmeir}, \bibinfo{person}{Geoff Webb}, \bibinfo{person}{Rob Hyndman}, {and} \bibinfo{person}{Pablo Montero-Manso}.} \bibinfo{year}{2021}\natexlab{}.
\newblock \bibinfo{booktitle}{\emph{Oikolab Weather Dataset}}.
\newblock \bibinfo{publisher}{Zenodo}.
\newblock
\urldef\tempurl%
\url{https://doi.org/10.5281/zenodo.5184708}
\showDOI{\tempurl}


\bibitem[\protect\citeauthoryear{Goldberger, Amaral, Glass, Hausdorff, Ivanov, Mark, Mietus, Moody, Peng, and Stanley}{Goldberger et~al\mbox{.}}{2000}]%
        {falldataset}
\bibfield{author}{\bibinfo{person}{Ary~L. Goldberger}, \bibinfo{person}{Luis A.~N. Amaral}, \bibinfo{person}{Leon Glass}, \bibinfo{person}{Jeffrey~M. Hausdorff}, \bibinfo{person}{Plamen~Ch. Ivanov}, \bibinfo{person}{Roger~G. Mark}, \bibinfo{person}{Joseph~E. Mietus}, \bibinfo{person}{George~B. Moody}, \bibinfo{person}{Chung-Kang Peng}, {and} \bibinfo{person}{H.~Eugene Stanley}.} \bibinfo{year}{2000}\natexlab{}.
\newblock \showarticletitle{PhysioBank, PhysioToolkit, and PhysioNet}.
\newblock \bibinfo{journal}{\emph{Circulation}} \bibinfo{volume}{101}, \bibinfo{number}{23} (\bibinfo{year}{2000}), \bibinfo{pages}{e215--e220}.
\newblock
\urldef\tempurl%
\url{https://doi.org/10.1161/01.CIR.101.23.e215}
\showDOI{\tempurl}
\showeprint{https://www.ahajournals.org/doi/pdf/10.1161/01.CIR.101.23.e215}


\bibitem[\protect\citeauthoryear{Ihlen, Weiss, Helbostad, and Hausdorff}{Ihlen et~al\mbox{.}}{2015}]%
        {fall}
\bibfield{author}{\bibinfo{person}{Espen A.~F. Ihlen}, \bibinfo{person}{Aner Weiss}, \bibinfo{person}{Jorunn~L. Helbostad}, {and} \bibinfo{person}{Jeffrey~M. Hausdorff}.} \bibinfo{year}{2015}\natexlab{}.
\newblock \showarticletitle{The Discriminant Value of Phase-Dependent Local Dynamic Stability of Daily Life Walking in Older Adult Community-Dwelling Fallers and Nonfallers}.
\newblock \bibinfo{journal}{\emph{BioMed Research International}} \bibinfo{volume}{2015}, \bibinfo{number}{1} (\bibinfo{year}{2015}), \bibinfo{pages}{402596}.
\newblock
\urldef\tempurl%
\url{https://doi.org/10.1155/2015/402596}
\showDOI{\tempurl}
\showeprint{https://onlinelibrary.wiley.com/doi/pdf/10.1155/2015/402596}


\bibitem[\protect\citeauthoryear{Inês~Silva and Henriques}{Inês~Silva and Henriques}{2020}]%
        {emdsilva}
\bibfield{author}{\bibinfo{person}{Maria Inês~Silva} {and} \bibinfo{person}{Roberto Henriques}.} \bibinfo{year}{2020}\natexlab{}.
\newblock \showarticletitle{Exploring time-series motifs through DTW-SOM}. In \bibinfo{booktitle}{\emph{IJCNN}}. \bibinfo{pages}{1--8}.
\newblock
\urldef\tempurl%
\url{https://doi.org/10.1109/IJCNN48605.2020.9207614}
\showDOI{\tempurl}


\bibitem[\protect\citeauthoryear{Law}{Law}{2019}]%
        {Law2019}
\bibfield{author}{\bibinfo{person}{Sean~M. Law}.} \bibinfo{year}{2019}\natexlab{}.
\newblock \showarticletitle{STUMPY: A Powerful and Scalable Python Library for Time Series Data Mining}.
\newblock \bibinfo{journal}{\emph{Journal of Open Source Software}} \bibinfo{volume}{4}, \bibinfo{number}{39} (\bibinfo{year}{2019}), \bibinfo{pages}{1504}.
\newblock
\urldef\tempurl%
\url{https://doi.org/10.21105/joss.01504}
\showDOI{\tempurl}


\bibitem[\protect\citeauthoryear{Leskovec, Rajaraman, and Ullman}{Leskovec et~al\mbox{.}}{2014}]%
        {2020mining}
\bibfield{author}{\bibinfo{person}{Jure Leskovec}, \bibinfo{person}{Anand Rajaraman}, {and} \bibinfo{person}{Jeffrey~David Ullman}.} \bibinfo{year}{2014}\natexlab{}.
\newblock \bibinfo{booktitle}{\emph{Mining of Massive Datasets} (\bibinfo{edition}{2nd} ed.)}.
\newblock \bibinfo{publisher}{Cambridge University Press}, \bibinfo{address}{USA}.
\newblock
\showISBNx{1107077230}


\bibitem[\protect\citeauthoryear{Lin, Keogh, Lonardi, and Chiu}{Lin et~al\mbox{.}}{2003}]%
        {sax}
\bibfield{author}{\bibinfo{person}{Jessica Lin}, \bibinfo{person}{Eamonn Keogh}, \bibinfo{person}{Stefano Lonardi}, {and} \bibinfo{person}{Bill Chiu}.} \bibinfo{year}{2003}\natexlab{}.
\newblock \showarticletitle{A symbolic representation of time series, with implications for streaming algorithms}. In \bibinfo{booktitle}{\emph{SIGMOD}}. \bibinfo{publisher}{ACM}, \bibinfo{address}{New York, NY, USA}, \bibinfo{pages}{2–11}.
\newblock
\showISBNx{9781450374224}
\urldef\tempurl%
\url{https://doi.org/10.1145/882082.882086}
\showDOI{\tempurl}


\bibitem[\protect\citeauthoryear{Liu, Li, Chen, Tan, Chen, and Zhou}{Liu et~al\mbox{.}}{2015}]%
        {health2015}
\bibfield{author}{\bibinfo{person}{Bo Liu}, \bibinfo{person}{Jianqiang Li}, \bibinfo{person}{Cheng Chen}, \bibinfo{person}{Wei Tan}, \bibinfo{person}{Qiang Chen}, {and} \bibinfo{person}{MengChu Zhou}.} \bibinfo{year}{2015}\natexlab{}.
\newblock \showarticletitle{Efficient Motif Discovery for Large-Scale Time Series in Healthcare}.
\newblock \bibinfo{journal}{\emph{IEEE Transactions on Industrial Informatics}} \bibinfo{volume}{11}, \bibinfo{number}{3} (\bibinfo{year}{2015}), \bibinfo{pages}{583--590}.
\newblock
\urldef\tempurl%
\url{https://doi.org/10.1109/TII.2015.2411226}
\showDOI{\tempurl}


\bibitem[\protect\citeauthoryear{Liu, Zhao, Liu, Wang, Li, Li, Lang, and Gu}{Liu et~al\mbox{.}}{2020}]%
        {PollutionMotifs}
\bibfield{author}{\bibinfo{person}{Bo Liu}, \bibinfo{person}{Huaipu Zhao}, \bibinfo{person}{Yinxing Liu}, \bibinfo{person}{Suyu Wang}, \bibinfo{person}{Jianqiang Li}, \bibinfo{person}{Yong Li}, \bibinfo{person}{Jianlei Lang}, {and} \bibinfo{person}{Rentao Gu}.} \bibinfo{year}{2020}\natexlab{}.
\newblock \showarticletitle{Discovering multi-dimensional motifs from multi-dimensional time series for air pollution control}.
\newblock \bibinfo{journal}{\emph{Concurrency and Computation: Practice and Experience}} \bibinfo{volume}{32}, \bibinfo{number}{11} (\bibinfo{year}{2020}), \bibinfo{pages}{e5645}.
\newblock
\urldef\tempurl%
\url{https://doi.org/10.1002/cpe.5645}
\showDOI{\tempurl}
\showeprint{https://onlinelibrary.wiley.com/doi/pdf/10.1002/cpe.5645}
\newblock
\shownote{e5645 CPE-19-1130.R1.}


\bibitem[\protect\citeauthoryear{Luzi, Lanzano, Felicetta, D’Amico, Russo, Sgobba, Pacor, and 5}{Luzi et~al\mbox{.}}{2020}]%
        {quake}
\bibfield{author}{\bibinfo{person}{Lucia Luzi}, \bibinfo{person}{Giovanni Lanzano}, \bibinfo{person}{Chiara Felicetta}, \bibinfo{person}{Maria~Chiara D’Amico}, \bibinfo{person}{Emilio Russo}, \bibinfo{person}{Sara Sgobba}, \bibinfo{person}{Francesca Pacor}, {and} \bibinfo{person}{ORFEUS Working~Group 5}.} \bibinfo{year}{2020}\natexlab{}.
\newblock \bibinfo{title}{Engineering Strong Motion Database (Version 2.0)}.
\newblock
\newblock
\urldef\tempurl%
\url{https://doi.org/10.13127/ESM.2}
\showDOI{\tempurl}


\bibitem[\protect\citeauthoryear{Minnen, Isbell, Essa, and Starner}{Minnen et~al\mbox{.}}{2007}]%
        {4470297}
\bibfield{author}{\bibinfo{person}{David Minnen}, \bibinfo{person}{Charles Isbell}, \bibinfo{person}{Irfan Essa}, {and} \bibinfo{person}{Thad Starner}.} \bibinfo{year}{2007}\natexlab{}.
\newblock \showarticletitle{Detecting Subdimensional Motifs: An Efficient Algorithm for Generalized Multivariate Pattern Discovery}. In \bibinfo{booktitle}{\emph{ICDM}}. \bibinfo{pages}{601--606}.
\newblock
\urldef\tempurl%
\url{https://doi.org/10.1109/ICDM.2007.52}
\showDOI{\tempurl}


\bibitem[\protect\citeauthoryear{M\"{o}rchen and Ultsch}{M\"{o}rchen and Ultsch}{2005}]%
        {persist2005}
\bibfield{author}{\bibinfo{person}{Fabian M\"{o}rchen} {and} \bibinfo{person}{Alfred Ultsch}.} \bibinfo{year}{2005}\natexlab{}.
\newblock \showarticletitle{Optimizing time series discretization for knowledge discovery}. In \bibinfo{booktitle}{\emph{SIGKDD}}. \bibinfo{publisher}{ACM}, \bibinfo{address}{New York, NY, USA}, \bibinfo{pages}{660–665}.
\newblock
\showISBNx{159593135X}
\urldef\tempurl%
\url{https://doi.org/10.1145/1081870.1081953}
\showDOI{\tempurl}


\bibitem[\protect\citeauthoryear{Murray, Stankovic, and Stankovic}{Murray et~al\mbox{.}}{2017}]%
        {PMID:28055033}
\bibfield{author}{\bibinfo{person}{David Murray}, \bibinfo{person}{Lina Stankovic}, {and} \bibinfo{person}{Vladimir Stankovic}.} \bibinfo{year}{2017}\natexlab{}.
\newblock \showarticletitle{An electrical load measurements dataset of United Kingdom households from a two-year longitudinal study}.
\newblock \bibinfo{journal}{\emph{Scientific data}}  \bibinfo{volume}{4} (\bibinfo{year}{2017}), \bibinfo{pages}{160122}.
\newblock
\showISSN{2052-4463}
\urldef\tempurl%
\url{https://doi.org/10.1038/sdata.2016.122}
\showDOI{\tempurl}


\bibitem[\protect\citeauthoryear{{NOAA Pacific Islands Fisheries Science Center}}{{NOAA Pacific Islands Fisheries Science Center}}{2021}]%
        {NOAA_PIPAN_2021}
\bibfield{author}{\bibinfo{person}{{NOAA Pacific Islands Fisheries Science Center}}.} \bibinfo{year}{2021}\natexlab{}.
\newblock \bibinfo{title}{PIPAN 10kHz Data}.
\newblock
\newblock
\urldef\tempurl%
\url{https://doi.org/10.25921/Z787-9Y54}
\showURL{%
\tempurl}
\newblock
\shownote{[Access date: 2024-10].}


\bibitem[\protect\citeauthoryear{Renard}{Renard}{2017}]%
        {renard:tel-01922186}
\bibfield{author}{\bibinfo{person}{Xavier Renard}.} \bibinfo{year}{2017}\natexlab{}.
\newblock \emph{\bibinfo{title}{{Time series representation for classification : a motif-based approach}}}.
\newblock Theses. \bibinfo{school}{{Universit{\'e} Pierre et Marie Curie - Paris VI}}.
\newblock
\urldef\tempurl%
\url{https://theses.hal.science/tel-01922186}
\showURL{%
\tempurl}


\bibitem[\protect\citeauthoryear{Rong, Yoon, Bergen, Elezabi, Bailis, Levis, and Beroza}{Rong et~al\mbox{.}}{2018}]%
        {LSHearthquake}
\bibfield{author}{\bibinfo{person}{Kexin Rong}, \bibinfo{person}{Clara~E. Yoon}, \bibinfo{person}{Karianne~J. Bergen}, \bibinfo{person}{Hashem Elezabi}, \bibinfo{person}{Peter Bailis}, \bibinfo{person}{Philip Levis}, {and} \bibinfo{person}{Gregory~C. Beroza}.} \bibinfo{year}{2018}\natexlab{}.
\newblock \showarticletitle{Locality-sensitive hashing for earthquake detection: a case study of scaling data-driven science}.
\newblock \bibinfo{journal}{\emph{Proc. VLDB Endow.}} \bibinfo{volume}{11}, \bibinfo{number}{11} (\bibinfo{date}{July} \bibinfo{year}{2018}), \bibinfo{pages}{1674–1687}.
\newblock
\showISSN{2150-8097}
\urldef\tempurl%
\url{https://doi.org/10.14778/3236187.3236214}
\showDOI{\tempurl}


\bibitem[\protect\citeauthoryear{Sant'Anna and Wickström}{Sant'Anna and Wickström}{2011}]%
        {sant2011symbolization}
\bibfield{author}{\bibinfo{person}{Anita Sant'Anna} {and} \bibinfo{person}{Nicholas Wickström}.} \bibinfo{year}{2011}\natexlab{}.
\newblock \showarticletitle{Symbolization of time-series: An evaluation of SAX, Persist, and ACA}. In \bibinfo{booktitle}{\emph{CISP}}, Vol.~\bibinfo{volume}{4}. \bibinfo{pages}{2223--2228}.
\newblock
\urldef\tempurl%
\url{https://doi.org/10.1109/CISP.2011.6100559}
\showDOI{\tempurl}


\bibitem[\protect\citeauthoryear{Sch{\"{a}}fer and Leser}{Sch{\"{a}}fer and Leser}{2024}]%
        {DBLP:journals/pvldb/SchaferL24}
\bibfield{author}{\bibinfo{person}{Patrick Sch{\"{a}}fer} {and} \bibinfo{person}{Ulf Leser}.} \bibinfo{year}{2024}\natexlab{}.
\newblock \showarticletitle{Discovering Leitmotifs in Multidimensional Time Series}.
\newblock \bibinfo{journal}{\emph{Proc. {VLDB} Endow.}} \bibinfo{volume}{18}, \bibinfo{number}{2} (\bibinfo{year}{2024}), \bibinfo{pages}{377--389}.
\newblock


\bibitem[\protect\citeauthoryear{Tanaka, Iwamoto, and Uehara}{Tanaka et~al\mbox{.}}{2005}]%
        {tanaka2005discovery}
\bibfield{author}{\bibinfo{person}{Yoshiki Tanaka}, \bibinfo{person}{Kazuhisa Iwamoto}, {and} \bibinfo{person}{Kuniaki Uehara}.} \bibinfo{year}{2005}\natexlab{}.
\newblock \showarticletitle{Discovery of Time-Series Motif from MultiDimensional Data Based on MDL Principle}.
\newblock \bibinfo{journal}{\emph{ML}}  \bibinfo{volume}{58} (\bibinfo{date}{02} \bibinfo{year}{2005}), \bibinfo{pages}{269--300}.
\newblock
\urldef\tempurl%
\url{https://doi.org/10.1007/s10994-005-5829-2}
\showDOI{\tempurl}


\bibitem[\protect\citeauthoryear{Vahdatpour, Amini, and Sarrafzadeh}{Vahdatpour et~al\mbox{.}}{2009}]%
        {ActivityDiscovery}
\bibfield{author}{\bibinfo{person}{Alireza Vahdatpour}, \bibinfo{person}{Navid Amini}, {and} \bibinfo{person}{Majid Sarrafzadeh}.} \bibinfo{year}{2009}\natexlab{}.
\newblock \showarticletitle{Toward unsupervised activity discovery using multi-dimensional motif detection in time series}. In \bibinfo{booktitle}{\emph{Proceedings of the 21st International Joint Conference on Artificial Intelligence}} (Pasadena, California, USA) \emph{(\bibinfo{series}{IJCAI'09})}. \bibinfo{publisher}{Morgan Kaufmann Publishers Inc.}, \bibinfo{address}{San Francisco, CA, USA}, \bibinfo{pages}{1261–1266}.
\newblock


\bibitem[\protect\citeauthoryear{Wang, Shen, Song, and Ji}{Wang et~al\mbox{.}}{2014}]%
        {wang2014hashingsimilaritysearchsurvey}
\bibfield{author}{\bibinfo{person}{Jingdong Wang}, \bibinfo{person}{Heng~Tao Shen}, \bibinfo{person}{Jingkuan Song}, {and} \bibinfo{person}{Jianqiu Ji}.} \bibinfo{year}{2014}\natexlab{}.
\newblock \bibinfo{title}{Hashing for Similarity Search: A Survey}.
\newblock
\newblock
\showeprint[arxiv]{1408.2927}~[cs.DS]
\urldef\tempurl%
\url{https://arxiv.org/abs/1408.2927}
\showURL{%
\tempurl}


\bibitem[\protect\citeauthoryear{Xiao, Cai, and Rajasekaran}{Xiao et~al\mbox{.}}{2019}]%
        {xiao2019edit}
\bibfield{author}{\bibinfo{person}{Peng Xiao}, \bibinfo{person}{Xingyu Cai}, {and} \bibinfo{person}{Sanguthevar Rajasekaran}.} \bibinfo{year}{2019}\natexlab{}.
\newblock \showarticletitle{Efficient Algorithms for Finding Edit-Distance Based Motifs}. In \bibinfo{booktitle}{\emph{Algorithms for Computational Biology}}, \bibfield{editor}{\bibinfo{person}{Ian Holmes}, \bibinfo{person}{Carlos Mart{\'i}n-Vide}, {and} \bibinfo{person}{Miguel~A. Vega-Rodr{\'i}guez}} (Eds.). \bibinfo{publisher}{Springer International Publishing}, \bibinfo{address}{Cham}, \bibinfo{pages}{212--223}.
\newblock
\showISBNx{978-3-030-18174-1}


\bibitem[\protect\citeauthoryear{Yeh, Kavantzas, and Keogh}{Yeh et~al\mbox{.}}{2017}]%
        {keoghMP}
\bibfield{author}{\bibinfo{person}{Chin-Chia~Michael Yeh}, \bibinfo{person}{Nickolas Kavantzas}, {and} \bibinfo{person}{Eamonn Keogh}.} \bibinfo{year}{2017}\natexlab{}.
\newblock \showarticletitle{Matrix Profile VI: Meaningful Multidimensional Motif Discovery}. In \bibinfo{booktitle}{\emph{ICDM}}. \bibinfo{pages}{565--574}.
\newblock
\urldef\tempurl%
\url{https://doi.org/10.1109/ICDM.2017.66}
\showDOI{\tempurl}


\bibitem[\protect\citeauthoryear{Yeh, Zhu, Ulanova, Begum, Ding, Dau, Silva, Mueen, and Keogh}{Yeh et~al\mbox{.}}{2016}]%
        {mp1}
\bibfield{author}{\bibinfo{person}{Chin-Chia~Michael Yeh}, \bibinfo{person}{Yan Zhu}, \bibinfo{person}{Liudmila Ulanova}, \bibinfo{person}{Nurjahan Begum}, \bibinfo{person}{Yifei Ding}, \bibinfo{person}{Hoang~Anh Dau}, \bibinfo{person}{Diego~Furtado Silva}, \bibinfo{person}{Abdullah Mueen}, {and} \bibinfo{person}{Eamonn Keogh}.} \bibinfo{year}{2016}\natexlab{}.
\newblock \showarticletitle{Matrix Profile I: All Pairs Similarity Joins for Time Series: A Unifying View That Includes Motifs, Discords and Shapelets}. In \bibinfo{booktitle}{\emph{2016 IEEE 16th International Conference on Data Mining (ICDM)}}. \bibinfo{pages}{1317--1322}.
\newblock
\urldef\tempurl%
\url{https://doi.org/10.1109/ICDM.2016.0179}
\showDOI{\tempurl}


\bibitem[\protect\citeauthoryear{Zhong and Mueen}{Zhong and Mueen}{2024}]%
        {zhong2024mass}
\bibfield{author}{\bibinfo{person}{Sheng Zhong} {and} \bibinfo{person}{Abdullah Mueen}.} \bibinfo{year}{2024}\natexlab{}.
\newblock \showarticletitle{MASS: distance profile of a query over a time series}.
\newblock \bibinfo{journal}{\emph{Data Min. Knowl. Discov.}} \bibinfo{volume}{38}, \bibinfo{number}{3} (\bibinfo{year}{2024}), \bibinfo{pages}{1466–1492}.
\newblock
\showISSN{1384-5810}
\urldef\tempurl%
\url{https://doi.org/10.1007/s10618-024-01005-2}
\showDOI{\tempurl}


\bibitem[\protect\citeauthoryear{Zhu, Zimmerman, Senobari, Yeh, Funning, Mueen, Brisk, and Keogh}{Zhu et~al\mbox{.}}{2016}]%
        {mp2}
\bibfield{author}{\bibinfo{person}{Yan Zhu}, \bibinfo{person}{Zachary Zimmerman}, \bibinfo{person}{Nader~Shakibay Senobari}, \bibinfo{person}{Chin-Chia~Michael Yeh}, \bibinfo{person}{Gareth Funning}, \bibinfo{person}{Abdullah Mueen}, \bibinfo{person}{Philip Brisk}, {and} \bibinfo{person}{Eamonn Keogh}.} \bibinfo{year}{2016}\natexlab{}.
\newblock \showarticletitle{Matrix Profile II: Exploiting a Novel Algorithm and GPUs to Break the One Hundred Million Barrier for Time Series Motifs and Joins}. In \bibinfo{booktitle}{\emph{2016 IEEE 16th International Conference on Data Mining (ICDM)}}. \bibinfo{pages}{739--748}.
\newblock
\urldef\tempurl%
\url{https://doi.org/10.1109/ICDM.2016.0085}
\showDOI{\tempurl}


\end{thebibliography}
\end{document}